\documentclass[letter, 10 pt, conference]{ieeeconf}  

\IEEEoverridecommandlockouts                              
\usepackage{graphics} 
\usepackage{epsfig} 
\usepackage{stmaryrd}
\usepackage{amsmath}

\usepackage{graphics,graphicx,psfrag,array,eepic}
\usepackage{multirow}
\usepackage{hhline}
\usepackage{mathtools}
\usepackage[strict]{changepage}
\usepackage{epsfig}
\usepackage{amsmath}
\usepackage{amssymb}

\usepackage{float}
\restylefloat{table}

\usepackage{multicol}
\usepackage{color}
\usepackage{balance}
\usepackage{cite}
\usepackage{booktabs} 
\usepackage{caption,booktabs}
\usepackage[tableposition=top]{caption}
\usepackage{subcaption}

\usepackage{lipsum}
\usepackage{comment}
\usepackage{amssymb}
\usepackage{graphicx}

\usepackage{comment}
\usepackage{bm}

\usepackage{algorithm}
\usepackage[noend]{algpseudocode}

\usepackage{mathtools, cuted}
\usepackage[thinc]{esdiff}

\usepackage{psfrag}
\usepackage{array}
\usepackage{latexsym,theorem}
\usepackage{amsfonts,amssymb,amsbsy}
\usepackage{tikz}
\usetikzlibrary{decorations.shapes,shapes.geometric,shadows,arrows,automata,positioning,calendar,mindmap,backgrounds,scopes,chains,er,3d,matrix,fit}
\usepackage{tikz-cd}
\usepackage{pgfplots}
\usetikzlibrary{intersections, pgfplots.fillbetween}

\newcommand{\R}{{\rm I\!R}}

\newcommand{\N}{{\rm I\!N}}

\newcommand{\cO}{\mathcal{O}}
\newcommand{\cS}{\mathcal{S}}
\newcommand{\cU}{\mathcal{U}}
\newcommand{\cX}{\mathcal{X}}

\newcommand{\cC}{\mathcal{C}}

\newcommand{\cK}{\mathcal{K}}

\newcommand{\cV}{\mathcal{V}}

\newcommand{\cY}{\mathcal{Y}}

\newtheorem{lemma}{Lemma}
\newtheorem{definition}{Definition}
\newtheorem{theorem}{Theorem}

\newtheorem{corollary}{Corollary}
\newtheorem{assumption}{Assumption}
\newtheorem{remark}{Remark}

\newtheorem{procedure}{Procedure}

\usepackage{stfloats}  
\allowdisplaybreaks

\usepackage{tikz}
\usepackage{pgfplots}
\usepackage{color}

\usepackage[british]{babel}
\usepackage{cite}
\usepackage{amsmath,amssymb,amsfonts}
\usepackage{graphicx}
\usepackage{textcomp}
\usepackage{mathrsfs}
\usepackage{url}
\usepackage{bm}
\usepackage{nicefrac}
\allowdisplaybreaks
\usepackage{alphalph,etoolbox}
\usepackage{mathtools,subcaption,float}
\usepackage{caption}
\usepackage{pgfplots}
\usepgfplotslibrary{colorbrewer}
\pgfplotsset{compat = 1.13} 
\usetikzlibrary{pgfplots.statistics, pgfplots.colorbrewer} 
\usepackage{pgfplotstable}

\allowdisplaybreaks

\IEEEoverridecommandlockouts 

\title{\LARGE \bf
Distributed Sequential Receding Horizon Control of Multi-Agent Systems under Recurring Signal Temporal Logic
}

\author{Eleftherios E. Vlahakis, Lars Lindemann, and Dimos V. Dimarogonas
\thanks{This work was supported by the Swedish
Research Council (VR), the Knut \& Alice Wallenberg Foundation (KAW), the Horizon Europe Grant SymAware and the ERC Consolidator Grant LEAFHOUND.}
\thanks{E. E. Vlahakis and D. V. Dimarogonas are with the Division of Decision and Control Systems, School of Electrical Engineering and Computer Science, KTH Royal Institute of Technology, 10044, Stockholm, Sweden. Email:
        {\tt\small\{vlahakis,dimos\}@kth.se}. L. Lindemann is with the Thomas Lord Department of Computer Science, Viterbi School of Engineering, University of Southern California,  Los Angeles, 90089, CA, USA. Email:
        {\tt\small llindema@usc.edu}.}
}

\begin{document}

\maketitle
\thispagestyle{empty}
\pagestyle{empty}

\begin{abstract}
We consider the synthesis problem of a multi-agent system under signal temporal logic (STL) specifications representing bounded-time tasks that need to be satisfied recurrently over an infinite horizon. Motivated by the limited approaches to handling recurring STL systematically, we tackle the infinite-horizon control problem with a receding horizon scheme equipped with additional STL constraints that introduce minimal complexity and a backward-reachability-based terminal condition that is straightforward to construct and ensures recursive feasibility. Subsequently, we decompose the global receding horizon optimization problem into agent-level programs the objectives of which are to minimize local cost functions subject to local and joint STL constraints. We propose a scheduling policy that allows individual agents to sequentially optimize their control actions while maintaining recursive feasibility. This results in a distributed strategy that can operate online as a model predictive controller. Last, we illustrate the effectiveness of our method via a multi-agent system example assigned a surveillance task.
\end{abstract}

\section{INTRODUCTION}

Various applications of multi-agent systems require agents to undertake recurring tasks that must be executed over an infinite time horizon. Examples 
include area coverage in surveillance systems, traffic management in intelligent transportation, smart harvesting in precision agriculture, and more. Signal temporal logic (STL) \cite{MalerSTL2004} is an important formalism for precisely defining, among others, a wide range of such tasks. STL is equipped with quantitative semantics \cite{Donze2010}, allowing one to capture the robustness of a formula, roughly measuring the distance from the space of falsifying system runs \cite{Fainekos2009}. Although formulation and verification of complex recurring multi-agent specifications are possible via STL, control synthesis from STL becomes 
challenging as the number of agents and the planning horizon grow.

A synthesis problem under recurring STL can be formulated as
\begin{subequations}\label{eq:inf_horizon_problem_intro}
\begin{align}
    \operatorname*{Minimize}_{\bm{u}(0),\; \bm{x}(0)}\;  &\sum_{t=0}^\infty \ell(x(t),u(t)) \label{eq:cost_intro} \\
     \text{subject to~} & x(t+1) = f(x(t),u(t)), \; x(0)=x_0,\label{eq:MAS_intro} \\
    & \bm{x}(0) \models \square_{[0,\infty)} \phi, \label{eq:recurring_STL_intro}  
\end{align}    
\end{subequations}
where $\ell:\R^n\times\R^m\to \R_{\geq 0}$ is a cost function, $\bm{u}(0) = (u(0),u(1),\ldots)$, $\bm{x}(0) = (x(0),x(1),\ldots)$ are the optimization variables denoting admissible input and state trajectories, respectively, with $\bm{x}(0)$ representing an infinite run of the system \eqref{eq:MAS_intro} that is required to satisfy $\phi$ in \eqref{eq:recurring_STL_intro} at every time step, with $\square$ denoting the \textit{always} operator and $\phi$ an STL formula. This problem is computationally challenging and hard to solve directly due to infinite memory requirements.

Intuitively, by iteratively solving a finite-horizon variant of \eqref{eq:inf_horizon_problem_intro} with a sufficiently large optimization horizon and tracking past computations, a feasible input trajectory $\bm{u}(0)$ can be constructed by collecting the first element of the optimal solution from each iteration. Such a receding horizon approach to addressing \eqref{eq:inf_horizon_problem_intro} recursively has been proposed in \cite{MurrayCDC2014} via mixed-integer programming (MIP) based on binary encoding of STL constraints. The soundness and completeness of this approach for STL formulas defined over linear predicates have resulted in approaches to reactive and robust model predictive control (MPC) of single-agent systems under STL \cite{Farahani2015, Sadraddini2015, RamanHSCC2015, GhoshHSCC2016}. However, due to the complexity increase caused by the use of binary variables \cite{Kurtz2022}, MIP-based synthesis methods are problematic when dealing with large or unbounded STL formulas 
and multi-agent settings.

To bypass complexity issues, the authors in \cite{LarsAutomatic2019} introduce discrete average space robustness, which allows the formulation of tractable optimization problems for control synthesis. Robustness metrics of STL that fit receding horizon schemes have also been proposed in \cite{Pant2017,MehdipourACC2019}. The appeal of a receding horizon approach relies on the validity that all iterations yield a solution. Although recursive feasibility is well explored for classical control objectives \cite{Mayne2000}, there exist only few results in recursive feasibility of receding horizon control under STL \cite{SadraddiniTAC2019, CharitidouTAC2023}.

Here, we address the control synthesis problem \eqref{eq:inf_horizon_problem_intro} of a multi-agent system subject to a recurring STL specification, which expresses a bounded-time task that needs to be satisfied recurrently over an infinite horizon, as per \eqref{eq:recurring_STL_intro}. We introduce a receding horizon scheme to attack \eqref{eq:inf_horizon_problem_intro} iteratively and propose additional constraints and terminal conditions that guarantee recursive feasibility enabling an MPC implementation. 
Subsequently, we decompose the receding horizon control problem into smaller agent-level problems and propose a distributed receding horizon scheme that iteratively addresses \eqref{eq:inf_horizon_problem_intro}, preserving recursive feasibility. This is facilitated by a scheduling policy that allows agents to sequentially optimize their actions, adhering to the couplings of the global STL formula. 

The remainder of the paper is organized as follows. Our MAS setup is introduced in Sec. \ref{sec:preliminaries}. The proposed receding horizon scheme and its distributed version are in Sec. \ref{sec:MAScontrol}. A numerical example is in Sec. \ref{sec:example}, whereas concluding remarks are discussed in Sec. \ref{sec:conclusion}.

\section{Preliminaries and Multi-Agent Setup}\label{sec:preliminaries}

\subsection{Notation}

The sets of real numbers and nonnegative integers are $\R$ and $\N$, 
respectively. Let $N\in \N$. Then, $\N_{[0,N]}=\{0,1,\ldots,N\}$. 
The transpose of $\xi$ is $\xi^\intercal$. 
The identity matrix is $I_n\in\R^{n\times n}$. The $j$th element of a vector $a$ is denoted as $(a)_j$.  
Let $x_{1},\ldots,x_{n}$ be vectors not necessarily of identical dimensions. 
Then, $x=(x_{1},\ldots,x_{n}) = [x_{1}^\intercal  \;\cdots\;x_{n}^\intercal  ]^\intercal$. 
We denote by $\bm{x}(a:b)=(x(a),\ldots,x(b))$ an aggregate vector consisting of $x(t)$, 
$t\in \N_{[a,b]}$, representing a trajectory. 
When it is clear from the context, we write $\bm{x}(t)$, omitting the endpoint. Let $x_i(t)$, for $t\in\N_{[0,N]}$ and $i\in\N_{[1,M]}$. 
Then, $\bm{x}(0:N)=(x(0),\ldots,x(N))$ denotes an aggregate trajectory when 
$x(t)=(x_1(t),\ldots,x_M(t))$, $t\in\N_{[0,N]}$. We denote by $x(t_\tau)$ the prediction of $x(t+\tau)$ carried out at time $t$. We also denote by $\bm{x}(t_{a:b})$ the predicted trajectory $(x(t_a),\ldots,x(t_b))$. The halfspace representation of a convex compact polytopic set is $\cY = \{y\mid G_y y \leq g_y\}$ with the inequalities applied elementwise. The cardinality of $\cV$ is $|\cV|$. The Minkowski sum of 
$S_1\subseteq \R^n$ and $S_2\subseteq \R^n$ is 
$S_1\oplus S_2=\{s_1+s_2 \mid s_1 \in S_1,\; s_2 \in S_2\}$. The Kronecker product of $A\in\R^{n\times m}$ and $B\in \R^{p\times q}$ is $A\otimes B \in\R^{np\times mq}$. The remainder of the division of $a$ by $b$ is $\mathrm{mod}(a,b)$.

\subsection{Signal Temporal Logic}\label{sec:STL}

We consider STL formulas with syntax 
\begin{equation}\label{eq:STL_syntax}
    \phi \coloneqq \top \mid \pi \mid  \lnot \phi  \mid  \phi_1 \wedge \phi_2 \mid \phi_1 U_{[a,b]}\phi_2,
\end{equation}
where $\pi\coloneqq (\mu(x) \geq 0)$ is a predicate, $\mu:\R^{n} \to \R$ being a predicate function, and $\phi_1$, $\phi_2$ are  STL formulas, which are specified recursively using predicates $\pi$, logical operators $\top$, $\neg$, and $\wedge$, and the 
\textit{until} temporal operator $U$, with $[a,b]\equiv \N_{[a,b]}$. Using the syntax in \eqref{eq:STL_syntax}, we may define $\lor$ (\textit{or}), $\lozenge$ (\textit{eventually}) and $\square$ (\textit{always}) 
operators, 
e.g., $\phi_1\lor \phi_2 = \neg (\neg \phi_1 \wedge \neg \phi_2)$, 
$\lozenge_{[a,b]}\phi = \top U_{[a,b]}\phi$, 
and $\square_{[a,b]}\phi = \lnot \lozenge_{[a,b]}\lnot \phi$.



We denote by $\bm{x}(t) \models \phi$ the satisfaction of the formula $\phi$ defined over the trajectory $\bm{x}(t)=(x(t),x(t+1),\ldots)$. 
The validity of an STL formula $\phi$ defined over a trajectory $\bm{x}(t)$ is verified as follows: 
$\bm{x}(t) \models \pi \Leftrightarrow \mu(x(t)) \geq 0$, 
$\bm{x}(t) \models \lnot \phi \Leftrightarrow \lnot(\bm{x}(t) \models \phi)$, 
$\bm{x}(t) \models \phi_1 \land \phi_2 \Leftrightarrow \bm{x}(t) \models \phi_1 \land \bm{x}(t) \models \phi_2$,  
$\bm{x}(t) \models \phi_1 U_{[a,b]}\phi_2 \Leftrightarrow \exists \tau\in  t\oplus \N_{[a,b]}$, 
s.t. $\bm{x}(\tau) \models \phi_2 \land \forall \tau'\in \N_{[t,\tau]}, \bm{x}(\tau') \models \phi_1$.

Based on the Boolean semantics introduced above, the horizon of a formula is recursively defined as \cite{MalerSTL2004}: $N^\pi = 0$, $N^{\lnot\phi} = N^{\phi}$, 
$N^{\phi_1\land \phi_2} = \max(N^{\phi_1}, N^{\phi_2})$, 
$N^{\phi_1\;U_{[a,b]}\phi_2} = b+ \max(N^{\phi_1}, N^{\phi_2})$.


STL is equipped with quantitative metrics for assessing the robustness of a formula \cite{Donze2010}. A scalar-valued function $\rho^\phi: \R^n\times\cdots\times \R^n \to \R$ indicates how robustly a trajectory $\bm{x}(t)$ satisfies a formula $\phi$. The robustness function is defined recursively as follows:
$\rho^\pi(\bm{x}(t)) = \mu(x(t))$, $\rho^{\lnot \phi}(\bm{x}(t)) = -\rho^{\phi}(\bm{x}(t))$,\\ 
$\rho^{\phi_1 \land \phi_2}(\bm{x}(t)) = \min(\rho^{\phi_1}(\bm{x}(t)),\rho^{\phi_2}(\bm{x}(t)))$, \\
${\tiny\rho^{\phi_1 U_{[a,b]}\phi_2}(\bm{x}(t)) = \max_{\tau\in t\oplus \N_{[a,b]}}\left( \min(Y_1(\tau),Y_2(\tau')) \right)}$, with $Y_1(\tau)=\rho^{\phi_1}(\bm{x}(\tau))$, $Y_2(\tau')=\min_{\tau'\in\N_{[t,\tau]}}\rho^{\phi_2}(\bm{x}(\tau'))$, $\pi$ being a predicate, and $\phi$, $\phi_1$, and $\phi_2$ being STL formulas.


\subsection{Multi-agent system}

We consider a multi-agent system (MAS) with discrete-time dynamics written as in \eqref{eq:MAS_intro}:
\begin{equation}\label{eq:MAS}
    x(t+1) = f(x(t),u(t)),
\end{equation}
where $x(t) = (x_1(t),\ldots,x_M(t))\in\cX\subset\R^n$ and $u(t) = (u_1(t),\ldots,u_M(t))\in\cU\subset\R^m$ are aggregate state and input vectors, respectively, with $x_i(t)\in\cX_i\subset\R^{n_i}$, $u_i(t)\in\cU_i\subset\R^{m_i}$, $\cX=\cX_1\times \cdots \times \cX_M$, $\cU=\cU_1\times \cdots \times \cU_M$, $M$ being the number of agents, and $x(0) = x_0$ is the initial condition. We assume that the agents are open-loop dynamically decoupled, i.e., $f(x(t),u(t)) = (f_1(x_1(t),u_1(t)),\ldots,f_M(x_M(t),u_M(t)))$, with $f(x(t),u(t)):\R^n \times \R^m \to \R^n$ and $f_i(x_i(t),u_i(t)):\R^{n_i} \times \R^{m_i} \to \R^{n_i}$, $i\in\N_{[1,M]}$. 

Let $\cV=\N_{[1,M]}$ be the set of agents' indices.
The MAS is tasked with a recurring STL specification $\psi=\square_{[0,\infty)}\phi$, that is, an STL specification, $\phi$, which needs to be satisfied at every time step $t\in\N$. The specification $\phi$ is a conjunctive formula, where each conjunct is an STL task involving a subset of agents $\nu\subseteq \cV$. We call $\nu$ a clique. Let $\cK_\phi$ be the set collecting all cliques induced by $\phi$. Then, the MAS is subject to the recurring STL specification:
\begin{subequations}\label{eq:global_psi}
    \begin{align}
        \psi &= \square_{[0,\infty)} \phi, \label{eq:psi_cliques} \\
        \phi & = \bigwedge_{\nu\in \cK_\phi} \phi_\nu, \label{eq:global_finite_phi}
    \end{align}
\end{subequations}
where a formula $\phi_\nu$ involves individual agents (represented by one-agent cliques), pairs of agents (represented by two-agent cliques), or groups of more than two agents (indicated by multi-agent cliques), with $|\nu|=1$, $|\nu|=2$, or $2<|\nu|\leq M$, respectively. Let $\pi\coloneqq(\mu(y)\geq 0)$ be a predicate in $\phi$, with $y\in\R^{n_y}$. The vector $y\in\R^{n_y}$ represents either an individual state vector, $x_i\in\R^{n_i}$, $i\in\N_{[1,M]}$, 
or an aggregate vector, $x_\nu\in\R^{n_\nu}$, collecting the state vectors of agents in the clique $\nu\in\cK_\phi$.


\begin{remark}
    Under the formulation \eqref{eq:global_psi}, there is no restriction on the number of subtasks that are associated with a clique $\nu\in \cK_\phi$. 
    In fact, $\phi_\nu$ is allowed to be any STL formula consistent with the syntax in \eqref{eq:STL_syntax}.
\end{remark}

\begin{assumption}\label{ass:indentical_horizons}
    We assume that the formulas $\phi_\nu$, $\nu\in \cK_\phi$, have identical horizons which we denote as $N$.  
\end{assumption}

\begin{remark}
    Assumption \ref{ass:indentical_horizons} is not restrictive due to the recurring structure of $\psi$. To see this, let, e.g., $\cK_\phi = \{\nu',\nu''\}$ and write $\psi = \square_{[0,\infty)} (\phi_{\nu'} \land  \phi_{\nu''})$. Let also $N^{\phi_{\nu'}}$ and $N^{\phi_{\nu''}}$ be the horizons of $\phi_{\nu'}$ and $\phi_{\nu''}$, respectively, with $N = \max\{N^{\phi_{\nu'}},N^{\phi_{\nu''}}\}$. Then, the satisfaction of $\psi$ may be enforced by satisfying $\hat{\psi} = \square_{[0,\infty)}((\square_{[0,N-N^{\phi_{\nu'}}]} \phi_{\nu'}) \land (\square_{[0,N-N^{\phi_{\nu''}}]} \phi_{\nu''}))$, where the horizon of $\square_{[0,N-N^{\phi_{\nu}}]} \phi_\nu$ is $N = N-N^{\phi_{\nu}}+N^{\phi_{\nu}}$, with $\nu=\{\nu',\nu''\}$. 
\end{remark}

\subsection{Decomposition of STL formula $\phi$}\label{sec:decomposition}
Let $\operatorname{cl}(i) = \{\nu \in \cK_\phi,\; \nu \ni i\}$ be the set of cliques that contain $i$, and assign a local STL formula to agent $i$ as
\begin{equation}\label{eq:varphi_i}
    \varphi_i = \bigwedge_{q\in \operatorname{cl}(i)}\phi_q.
\end{equation}
Let $q\in \operatorname{cl}(i)$, where $q=(i_1,\ldots,i,\ldots,i_{q'})$, with $q'+1=|q|$. Let a trajectory $\bm{x}_{q}(t)=(x_{q}(t),x_{q}(t+1),\ldots)$, where $x_{q}(t)=(x_{i_1}(t),\ldots,x_i(t),\ldots,x_{i_{q'}}(t))$, with $t\in\N$, and the order $i_1<\cdots<i<\cdots<i_{q'}$ being specified by the lexicographic ordering of the set $\cV=\N_{[1,M]}$. Then, we denote by $\bm{x}_{q}(t)\models \phi_{q}$ the satisfaction of $\phi_{q}$ verified over the aggregate trajectory $\bm{x}_{q}(t)$.

\begin{remark}
    Comparing \eqref{eq:global_finite_phi} and \eqref{eq:varphi_i}, one may easily verify that $\bm{x}(t)\models \bigwedge_{i\in\cV}\varphi_i$ implies $\bm{x}(t)\models \phi$. 
\end{remark}



\section{Main Results}\label{sec:MAScontrol}

\subsection{Receding horizon control}

In this section, we propose a receding-horizon variant of \eqref{eq:inf_horizon_problem_intro} by incorporating additional constraints and terminal conditions that enforce trajectories to exhibit a recurring pattern. Next, we recall the following definitions.


\begin{definition}
    For the system \eqref{eq:MAS}, we denote the precursor set to a set $\cS\subset \cX$ as follows
    \begin{equation}
        \text{Pre}(\cS) = \{ x\in \R^n \mid \exists u\in \cU\; \text{s.t.\;} f(x,u) \in \cS \}.
    \end{equation}
\end{definition}

\begin{definition}[one-Step Controllable Set]
    For the admissible set $\cX$ and a given target set $\cS\subseteq \cX$ the one-step controllable set $\cC_1(\cS)$ of the system \eqref{eq:MAS} is defined as
    \begin{equation}\label{eq:controllable_set_C1}
        \cC_1(\cS) = \text{Pre}(\cS)\cap \cX.
    \end{equation}
\end{definition}
The set $\cC_1(\cS)$ collects all the states of \eqref{eq:MAS} that can be driven to the target set $\cS$ in one step. We will now introduce a receding horizon version of \eqref{eq:inf_horizon_problem_intro} and state its recursive feasibility property. At time $t\in \N$, 
\begin{subequations}\label{eq:receding_horizon_relax_MAS}
    \begin{align}
     \operatorname*{Min.}_{\substack{\bm{u}(t_{0:N})\\  \bm{x}(t_{0:N+1})}}&  \sum_{i=1}^{M}\big( \sum_{k=0}^N \ell_i(x_i(t_k),u_i(t_k)) +V_i(x_i(t_{N+1}))\big) \\
     \text{s.t.\;} x&(t_{k+1}) = f(x(t_k),u(t_k)), \; k \in\N_{[0,N]},\\
     (&\bm{x}(t-N+1:t-1),\bm{x}(t_{0:1})) \models \phi, \label{eq:history1_MAS}\\
    & \vdots \nonumber\\
     (&x(t-1),\bm{x}(t_{0:N-1})) \models \phi, \label{eq:history4_MAS}\\
     \bm{x}&(t_{0:N}) \models \phi,  \label{eq:constraint_phi_MAS}\\ 
    \bm{x}&(t_{1:N+1})\models \phi, \label{eq:con1_MAS}\\
    & \vdots \nonumber\\
     (&x(t_{N+1}),\bm{x}(t_{1:N}))\models\phi, \label{eq:con5_MAS}\\
    (&\bm{x}(t_{1:N}),x(t_0))\models \phi, \label{eq:con6_MAS}\\
    & \vdots \nonumber\\
     (&x(t_N),\bm{x}(t_{0:N-1}))\models\phi, \label{eq:con7_MAS}\\
    x&(t_0) = x(t), \\
    x&(t_k), x(t_{N+1}) \in  \cX,\;  u(t_k) \in \cU,\; k\in\N_{[0,N]}, \label{eq:conX_MAS}\\
     x&(t_N) \in \cC_1(x(t)),\label{eq:varying_terminal_constr_MAS}
    \end{align}
\end{subequations}
where $\ell_i(x_i(t_k),u_i(t_k))$, $V_i(x(t_{N+1}))$, with $\ell_i:\R^{n_i}\times\R^{m_i}\to \R_{\geq 0}$, $V_i:\R^{n_i}\to \R_{\geq 0}$, represent agent-level decomposition of the cost function in \eqref{eq:cost_intro}, $x(t) = (x_1(t), \ldots,x_M(t))$, $\phi = \bigwedge_{\nu\in \cK_\phi}\phi_\nu$, $N$ is the horizon of $\phi$, and $\cC_1(x(t))$ in \eqref{eq:varying_terminal_constr_MAS} is the one-step controllable set of system \eqref{eq:MAS}. The constraints in \eqref{eq:history1_MAS}-\eqref{eq:history4_MAS} consider trajectories composed of historical data, the current state and predicted states, whereas the one in \eqref{eq:constraint_phi_MAS} considers a trajectory originating at $x(t)$. The constraints in \eqref{eq:con1_MAS}-\eqref{eq:con5_MAS} enforce that any trajectory consisting of $N+1$ successive elements of $(\bm{x}(t_{1:N+1}),\bm{x}(t_{1:N}))$ satisfies $\phi$, whereas the ones in \eqref{eq:con6_MAS}-\eqref{eq:con7_MAS} enforce that any trajectory consisting of $N+1$ successive elements of $(\bm{x}(t_{1:N}),\bm{x}(t_{0:N-1}))$ satisfies $\phi$. Last, the terminal condition in \eqref{eq:varying_terminal_constr_MAS} enforces recursive feasibility (see Theorem \ref{thm:recurs_feasibility}), and the constraints in \eqref{eq:conX_MAS} enforce safety and input capacity constraints. 

\begin{remark}
    Given linear predicate functions, if binary encoding is employed as in \cite{MurrayCDC2014}, the constraints in \eqref{eq:con1_MAS}-\eqref{eq:con7_MAS} require $2N$ additional binary variables. 
\end{remark}


\begin{theorem}\label{thm:recurs_feasibility}
    Let the optimization problem \eqref{eq:receding_horizon_relax_MAS} be feasible at $t=0$. Then, it is feasible for all $t\geq 0$.
\end{theorem}
    
\begin{proof}
See Appendix I.
\end{proof}

An immediate consequence of Theorem \ref{thm:recurs_feasibility} is stated next.
\begin{corollary}
Let problem \eqref{eq:receding_horizon_relax_MAS} be feasible at time $t=0$ and denote a (not necessarily optimal) solution to it at time $t\in\N$ by $\bm{u}^\ast(t_{0:N})=(u^\ast(t_0),\ldots,u^\ast(t_N))$. Then, by applying $u(t) = u^\ast(t_0)$ to \eqref{eq:MAS}, $\bm{x}(t)\models \phi$ for all $t\in\N$, or, $\bm{x}(0) = (x(0),x(1),\ldots)$, satisfies $\square_{[0,\infty)}\phi$, that is, $\bm{x}(0) \models \psi$.  
\end{corollary}

\begin{remark}
    Constructing $\cC_1(x(t))$ is challenging for nonlinear systems. As a remedy, we can replace the terminal condition in \eqref{eq:varying_terminal_constr_MAS} with $x(t_{N+1}) = x(t)$ and eliminate the redundant constraints in \eqref{eq:con6_MAS}-\eqref{eq:con7_MAS}. This modification comes at the cost of a more restricted single-point reachability condition, but it maintains recursive feasibility. 
\end{remark}


\subsection{Distributed receding horizon control}

Next, we decompose the global control problem \eqref{eq:receding_horizon_relax_MAS} into agent-level problems while preserving recursive feasibility. Our approach relies on the following assumption.

\begin{assumption}\label{ass:optim_MAS_feasible_at_t0}
    A feasible solution to the problem \eqref{eq:receding_horizon_relax_MAS} for $t=0$ is known.
\end{assumption}

We are now ready to formulate the following optimization problems for all $i\in\N_{[1,M]}$:
\begin{subequations}\label{eq:receding_horizon_relax_node_i}
    \begin{align}
    &\operatorname*{Min.}_{\substack{\bm{u}_i(t_{0:N})\\  \bm{x}_i(t_{0:N+1})}}
    \sum_{k=0}^N \big( \ell_i(x_i(t_k),u_i(t_k)) \big) +V_i(x_i(t_{N+1})) \\
     &\text{s.t.~}  x_i(t_{k+1}) = f_i(x_i(t_k),u_i(t_k)), \; k \in\N_{[0,N]}, \\
    & \hspace{15pt} (\bm{x}_q(t-N+1:t-1),\bm{x}_q(t_{0:1})) \models \phi_q, \label{eq:history1_node_i}\\
    & \hspace{18pt} \vdots \nonumber\\
    & \hspace{15pt} (x_q(t-1),\bm{x}_q(t_{0:N-1})) \models \phi_q, \label{eq:history4_node_i}\\
    & \hspace{15pt}\bm{x}_q(t_{0:N}) \models \phi_q,  \label{eq:constraint_phi_node_i}\\ 
     &\hspace{15pt} \bm{x}_q(t_{1:N+1})\models \phi_q, \label{eq:con1_node_i}\\
    & \hspace{18pt} \vdots \nonumber\\
    & \hspace{15pt} (x_q(t_{N+1}),\bm{x}_q(t_{1:N}))\models\phi_q, \label{eq:con5_node_i}\\
    & \hspace{15pt} (\bm{x}_q(t_{1:N}),x_q(t_0))\models \phi_q, \label{eq:con6_node_i}\\
    & \hspace{18pt} \vdots \nonumber\\
    & \hspace{15pt} (x_q(t_N),\bm{x}_q(t_{0:N-1}))\models\phi_q, \label{eq:con7_node_i}\\
    & \hspace{15pt} \forall q \in \operatorname{cl}(i), \nonumber\\
    & \hspace{15pt} x_i(t_0) = x_i(t), \\
    & \hspace{15pt} x_i(t_k), x_i(t_{N+1}) \in \cX_i,\;u_i(t_k) \in \cU_i,\; k\in\N_{[0,N]}, \\
    & \hspace{15pt} x_i(t_N) \in \cC_1(x_i(t)).\label{eq:varying_terminal_constr_node_i}
    \end{align}
\end{subequations}

The main idea for introducing \eqref{eq:receding_horizon_relax_node_i} is to attack the global problem, \eqref{eq:receding_horizon_relax_MAS}, by solving smaller agent-level problems in a distributed fashion, while preserving recursive feasibility.

\begin{lemma}\label{lem:receding_horizon_prob_node_i_feasibility}
    Let Assumption \ref{ass:optim_MAS_feasible_at_t0} hold. Then, there is always an admissible input sequence that renders \eqref{eq:receding_horizon_relax_node_i} feasible for all $t\in\N$ and all agents $i\in\N_{[1,M]}$. 
\end{lemma}
\begin{proof}
    At $t=0$, let $\bm{u}(0_{0:N}) = (u(0_0),\ldots,u(0_N))$, $\bm{x}(0_{0:N+1})=(x(0_0),\ldots,x(0_{N+1}))$ be a known feasible solution to \eqref{eq:receding_horizon_relax_MAS}, with $u(0_\tau) =(u_1(0_\tau),\ldots,u_M(0_\tau))$, $\tau\in\N_{[0,N]}$, $x(0_{\tau'}) =(x_1(0_{\tau'}),\ldots,x_M(0_{\tau'}))$,  $\tau'\in\N_{[0,N+1]}$. By feasibility at $t=0$, there exist inputs $\hat{u}_i(N)$ such that $f_i(x_i(0_N),\hat{u}_i(N)) = x_i(0_0)$, $i\in\N_{[1,M]}$. One can, then, see that there always exist $N+1$ successive elements in the sequence $(u_i(0_0),u_i(0_1),\ldots,\hat{u}_i(N),u_i(0_0),u_i(0_1),\ldots,\hat{u}_i(N))$ that lie in the feasible domain of \eqref{eq:receding_horizon_relax_node_i} for all $i\in\N_{[1,M]}$. 
\end{proof}

Lemma \ref{lem:receding_horizon_prob_node_i_feasibility} states that the recursive feasibility property of the optimization problem \eqref{eq:receding_horizon_relax_MAS} is preserved even if we decompose it into $M$ agent-level problems \eqref{eq:receding_horizon_relax_node_i}. However, simultaneously solving \eqref{eq:receding_horizon_relax_node_i} for agents belonging to the same cliques, $\nu\in\cK_\phi$, becomes problematic, due to the difficulty to address coupling constraints in parallel. As a remedy, we propose a scheduling algorithm wherein only a subset of agents can optimize their control inputs at a given time step, whereas the remaining peers regress a feasible solution from the preceding time step. The implementation of this sequential scheduling gives rise to a distributed receding horizon control policy, while ensuring recursive feasibility. 

\subsection{Distributed sequential model predictive control}

Let $\mathcal{O} = \{q \in \mathcal{V} \mid \forall q' \in \mathcal{V}, \; \operatorname{cl}(q) \cap \operatorname{cl}(q')=\emptyset \}$ be a set consisting of indices of agents that are allowed to solve their local problems \eqref{eq:receding_horizon_relax_node_i} in parallel. Obviously, the set $\cO$ is not unique. We denote by $\cO(t)\subset \cV$ the set of agents selected at time $t\in\N$. We construct such a set as follows. At each time $t\in\N$, we initialize $\cO(t)$ by $v(t)\in \cV$, where $v(t) = \mathrm{mod}(v(t-1),M)+1$, with $M=|\cV|$. The time dependence in $v(t)$ indicates the choice of an agent from the set $\cV$ at time $t$. For example, let $\cV=\{1,2,3\}$. Then, we initialize $\cO(0),\cO(1),\ldots$ by $v(0)=1$, $v(1)=2$, $v(2)=3$, $v(3)=1$, and so forth, respectively. After initialization, by going through the remaining $M-1$ indices in $\cV$, following a lexicographic order, we add an index $q\in\cV$ to $\cO(t)$ if $\operatorname{cl}(q)\cap \operatorname{cl}(q')=\emptyset$ for all $q'\in\cO(t)$. 


\begin{procedure}\label{proc:Ot}
    At time $t\in\N$, if $i\in\cO(t)$, an input sequence, $\bm{u}_i(t_{0:N})$, is produced by solving Problem \eqref{eq:receding_horizon_relax_node_i} (potentially optimally). If $i\notin \cO(t)$, an input sequence, $\bm{u}_i(t_{0:N})$, is constructed by the preceding time step, which is feasible for Problem \eqref{eq:receding_horizon_relax_node_i} (cf. Lemma \ref{lem:receding_horizon_prob_node_i_feasibility}). 
\end{procedure}

Based on the sequential scheduling induced by the construction of the set $\cO(t)$, we are in a position to state the following result and its immediate consequence. 
\begin{theorem}\label{thm:distributed_recursive_feasibility}
    Consider Procedure \ref{proc:Ot} and let Assumption \ref{ass:optim_MAS_feasible_at_t0} hold. Then, \eqref{eq:receding_horizon_relax_node_i} is recursively feasible for all $i\in \cV$. 
\end{theorem}
\begin{proof}
    The result relies on the construction of the set $\cO(t)$. Let $\bm{u}_i(t_{0:N})$, $\bm{x}_i(t_{0:N+1})$ be feasible input and state trajectories, respectively, obtained at time $t\in\N$ by solving \eqref{eq:receding_horizon_relax_node_i} for $i\in\cO(t)$. Note that $\bm{x}_i(t_{0:N+1})$ does not affect the feasibility of any Problem \eqref{eq:receding_horizon_relax_node_i} corresponding to agent $j\notin \cO(t)$, with $\operatorname{cl}(j)\cap \operatorname{cl}(i)\neq \emptyset$. For such agents, a feasible input sequence for time $t$ may be constructed by the feasible solution obtained in time step $t-1$ (see proof of Lemma \ref{lem:receding_horizon_prob_node_i_feasibility}), which is known to all agents $i\in\cO(t)$ at time $t$. Finally, recursive feasibility follows by Lemma \ref{lem:receding_horizon_prob_node_i_feasibility}.
\end{proof}

\begin{corollary}
    Consider Procedure \ref{proc:Ot} and let Assumption \ref{ass:optim_MAS_feasible_at_t0} hold. Then, at each time $t\in\N$, there exist admissible local inputs $u_i(t_0)$, $i\in\N_{[1,M]}$  synthesizing a multi-agent input $u(t_0)=(u_1(t_0),\ldots,u_M(t_0))$. These inputs, when applied to \eqref{eq:MAS}, yield a multi-agent state trajectory $\bm{x}(t)$ that satisfies $\phi$ for all $t\in\N$, that is, $\bm{x}(0) \models \psi$.  
\end{corollary}

A distributed MPC scheme based on the scheduling algorithm induced by the set $\cO(t)$ and feasible solutions of problems \eqref{eq:receding_horizon_relax_node_i}, denoted as $\bm{u}_i^\ast(t_{0:N})$, $\bm{x}_i^\ast(t_{0:N+1})$, is implemented as follows. At time $t\in\N$:

\begin{enumerate}
    \item[(S1)] Agent $i$ measures its local state and exchanges predictions on its receding-horizon state trajectory with agents $j\in\cV$, where $\operatorname{cl}(j)\cap \operatorname{cl}(i)\neq \emptyset$, as computed in the previous time step. 
    \item[(S2)] If $i\in\cO(t)$, agent $i$ solves \eqref{eq:receding_horizon_relax_node_i} obtaining $\bm{u}_i^\ast(t_{0:N})$, and implements its first component, i.e., $u_i(t) = u_i^\ast(t_0)$. If $i\notin \cO(t)$, $u_i(t) = u_i(t'_1)$, with $t'=t-1$. 
    \item[(S3)] At time $t+1$, agent $i$ repeats steps (S1) and (S2). 
\end{enumerate}


\section{Numerical Example}\label{sec:example}

We exemplify our distributed control strategy using the MAS example considered in \cite[Sec. 6]{LarsAutomatic2019}, which we augment by collision avoidance among agents. The MAS consists of three identical agents, with indices, $\cV=\{1,2,3\}$, with global dynamics, $x(t+1) = Ax(t)+Bu(t)$, where $x(t) = (x_1(t),x_2(t),x_3(t))\in\R^{12}$, $u(t)=(u_1(t),u_2(t),u_3(t))\in\R^{6}$, with $x_i(t)\in\R^4$, $u_i(t)\in\R^2$, $i\in\N_{[1,3]}$, and $A = I_6\otimes \begin{bmatrix}
    1 & 0.1\\0 & 1
\end{bmatrix}$, $B = I_6\otimes \begin{bmatrix}
    0.005\\1
\end{bmatrix}$. The vector $((x_i)_1,(x_i)_3)\in\cX$ denotes the position of agent $i$ in $\cX$, where $\cX = \{\xi\in\R^2 \mid \begin{bmatrix}
    I_2&-I_2
\end{bmatrix}^\intercal \xi \leq (10,10,0,0)\}$, the vector $((x_i)_2,(x_i)_4)$ collects the respective velocities, and the vector $u_i$ collects the accelerations of agent $i$, where $-20\leq (u_i)_j \leq 20$, with $j\in\N_{[1,2]}$, $i\in\N_{[1,3]}$.

The MAS is tasked with surveilling the area $\cX\subset \R^2_{\geq 0}$. This specification is expressed through a recurring STL formula $\psi = \square_{[0,\infty)}\phi$, where $\phi = \phi^1\wedge\phi^2\wedge\phi^3$, with $\phi^1=\phi^{11}\wedge\phi^{12}\wedge\phi^{13}$, $\phi^2=\phi^{21}\wedge\phi^{22}\wedge\phi^{23}$, $\phi^3=\phi^{31}\wedge\phi^{32}$. Specifically, $\phi^1$ defines the admissible workspace, $\phi^2$ specifies the surveillance task, and $\phi^3$ specifies collision avoidance, with $\phi^{11}=\square_{[0,30]}\|((x_1)_1,(x_1)_3)-(5,5)\|_\infty\leq 5$, $\phi^{12}=\square_{[0,30]}\|((x_2)_1,(x_2)_3)-(5,5)\|_\infty\leq 5$,\\ $\phi^{13}=\square_{[0,30]}\|((x_3)_1,(x_3)_3)-(5,5)\|_\infty\leq 5$,\\
$\phi^{21} = \left(\lozenge_{[10,30]}\|((x_1)_1,(x_1)_3)-(5,9)\|_\infty\leq 1\right)\wedge\left(\lozenge_{[10,30]}\|((x_1)_1,(x_1)_3)-(1,5)\|_\infty\leq 1\right)$,\\ 
$\phi^{22} = \left(\lozenge_{[10,30]}\|((x_2)_1,(x_2)_3)-(8,8)\|_\infty\leq 1\right)\wedge\left(\lozenge_{[10,30]}\|((x_2)_1,(x_2)_3)-(2,2)\|_\infty\leq 1\right)$,\\
$\phi^{23} = \left(\lozenge_{[10,30]}\|((x_3)_1,(x_3)_3)-(9,5)\|_\infty\leq 1\right)\wedge\left(\lozenge_{[10,30]}\|((x_3)_1,(x_3)_3)-(5,1)\|_\infty\leq 1\right)$,\\
$\phi^{31}=\square_{[0,30]}\|((x_1)_1,(x_1)_3)-((x_2)_1,(x_2)_3)\|_\infty\geq 0.1$, \\
$\phi^{32}=\square_{[0,30]}\|((x_2)_1,(x_2)_3)-((x_3)_1,(x_3)_3)\|_\infty\geq 0.1$.

To formulate local formulas, $\varphi_i$, as defined in \eqref{eq:varphi_i}, we identify the set of cliques induced by the formula $\phi$ as $\cK_\phi=\{1,2,3,(1,2),(2,3)\}$, and collect the cliques associated with each agent in $\operatorname{cl}(1)=\{1,(1,2)\}$, $\operatorname{cl}(2)=\{2,(1,2),(2,3)\}$, and $\operatorname{cl}(3)=\{3,(2,3)\}$. These yield local formulas $\varphi_1=\phi_1\wedge\phi_{12}$, $\varphi_2=\phi_2\wedge\phi_{12}\wedge\phi_{23}$, and $\varphi_3=\phi_3\wedge\phi_{23}$, where, $\phi_1 = \phi^{11}\wedge\phi^{21}$, $\phi_{12}=\phi^{31}$, $\phi_2 = \phi^{21}\wedge\phi^{22}$, $\phi_{23}=\phi^{32}$, and $\phi_3 = \phi^{13}\wedge\phi^{23}$. Next, based on the local formulas, $\varphi_i$, $i\in\N_{[1,3]}$, we formulate local problems \eqref{eq:receding_horizon_relax_node_i}, where we express infinity-norm-based constraints as conjunctions and disjunctions of linear inequalities and formulate the entire STL constraints in $\varphi_i$ using binary encoding \cite{MurrayCDC2014}.


In our simulations, we select $\ell = \|u(t)\|_1$, $\ell_i=\|u_i(t)\|_1$, and $V_i=0$ as global and local performance criteria, respectively. We initiate the process by solving \eqref{eq:receding_horizon_relax_MAS} for one iteration and subsequently execute the distributed receding horizon procedure (S1)-(S3), as detailed in the previous section. Following the initial iteration, an optimization is performed by either agent 2 alone or agents 1 and 3, determined by the set $\cO(t) \in \{\cO_1,\cO_2\}$, for $t\in\N$, where $\cO_1=\{2\}$ and $\cO_2=\{1,3\}$. The underlying mixed-integer linear programs were formulated using the YALMIP toolbox in MATLAB \cite{Yalmip} and solved using the GUROBI solver \cite{gurobi}. 

Starting from the chosen initial conditions, our distributed control strategy guides the MAS along a trajectory that satisfies the recurring STL formula $\psi$, while respecting input constraints, as shown in Fig. \ref{fig:main}. Computational aspects for the initial centralized iteration and the subsequent distributed iterations are depicted in Fig. \ref{fig:boxplots}. To assess the robustness of our method, we conducted simulations by introducing velocity perturbations to the MAS for an 80-time-step duration. Despite these perturbations, the MAS effectively remains within the permissible workspace, successfully fulfilling its surveillance task, as depicted in Fig. \ref{fig:main2}. To verify both simulations, we employed the robustness function detailed in Sec. \ref{sec:STL}, which remained nonnegative for most of the simulation period. In addition, we performed simulations in which we omitted the proposed terminal constraint. The outcomes reveal that, given the chosen input constraints, the MAS cannot successfully accomplish the surveillance task. This underscores the key role of the terminal condition in our control design.

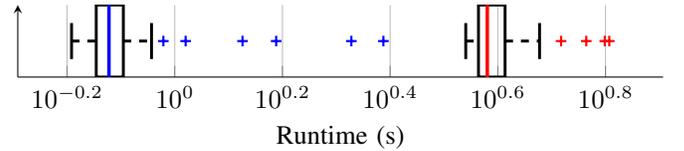
\begin{figure}[ht]
  \centering
\begin{tikzpicture}
\begin{axis}[%
width=4.in,
height=1.in,
at={(0in,0.in)}, 
boxplot/draw direction = x, 
x axis line style = {opacity=0},
y axis line style = {opacity=0},
axis x line* = bottom,
axis y line = left,
enlarge x limits,
xmajorgrids,
yticklabel style = {align=center, font=\small},
yticklabels = {?},
ytick style = {draw=none}, 
xlabel = {Runtime (s)},
xmode=log
]
\addplot [color=black, dashed, line width=1.1pt, forget plot]
  table[row sep=crcr]{%
4.1064146	1\\
4.7586932	1\\
};
\addplot [color=black, line width=1.1pt, dashed, forget plot]
  table[row sep=crcr]{%
3.4706864	1\\
3.6653605	1\\
};
\addplot [color=black, line width=1.1pt, forget plot]
  table[row sep=crcr]{%
4.7586932	0.9625\\
4.7586932	1.0375\\
};
\addplot [color=black, line width=1.1pt, forget plot]
  table[row sep=crcr]{%
3.4706864	0.9625\\
3.4706864	1.0375\\
};
\addplot [color=black,line width=1.1pt, forget plot]
  table[row sep=crcr]{%
3.6653605	0.925\\
4.1064146	0.925\\
4.1064146	1.075\\
3.6653605	1.075\\
3.6653605	0.925\\
};
\addplot [color=red, line width=1.3pt, forget plot]
  table[row sep=crcr]{%
3.80573735	0.925\\
3.80573735	1.075\\
};
\addplot [color=black, only marks, mark=+,  line width=0.8pt, mark options={solid, draw=red}, forget plot]
  table[row sep=crcr]{%
5.2174031	1\\
5.8119277	1\\
6.2887858	1\\
6.4122631	1\\
};
\addplot [color=black, dashed, line width=1.1pt, forget plot]
  table[row sep=crcr]{%
0.802655	1\\
0.9055119	1\\
};
\addplot [color=black, dashed, line width=1.1pt, forget plot]
  table[row sep=crcr]{%
0.6432866	1\\
0.71489	1\\
};
\addplot [color=black, line width=1.1pt, forget plot]
  table[row sep=crcr]{%
0.9055119	0.9625\\
0.9055119	1.0375\\
};
\addplot [color=black, line width=1.1pt, forget plot]
  table[row sep=crcr]{%
0.6432866	0.9625\\
0.6432866	1.0375\\
};
\addplot [color=black, line width=1.1pt, forget plot]
  table[row sep=crcr]{%
0.71489	0.925\\
0.802655	0.925\\
0.802655	1.075\\
0.71489	1.075\\
0.71489	0.925\\
};
\addplot [color=blue, line width=1.3pt, forget plot]
  table[row sep=crcr]{%
0.7548692	0.925\\
0.7548692	1.075\\
};
\addplot [color=black, only marks, line width=0.8pt,  mark=+, mark options={solid, draw=blue}, forget plot]
  table[row sep=crcr]{%
0.9525171	1\\
1.0476409	1\\
1.3362613	1\\
1.5430826	1\\
2.1280186	1\\
2.4403274	1\\
};
\end{axis}

\end{tikzpicture}

\caption{Runtime for computations on an i7-1185G7 CPU at 3.00GHz. (Left) Distributed scheme. (Right) Centralized scheme.}
\label{fig:boxplots}
\end{figure}

\begin{figure*} 
  \centering
  \begin{subfigure}{0.245\textwidth} 
    \definecolor{mycolor1}{rgb}{0.00000,0.44700,0.74100}%
\definecolor{mycolor2}{rgb}{0.85000,0.32500,0.09800}%
\definecolor{mycolor3}{rgb}{0.92900,0.69400,0.12500}%
\begin{tikzpicture}

\begin{axis}[%
width=1.3in,
height=1.16in,
at={(0.in,0.in)},
scale only axis,
xmin=0,
xmax=130,
xlabel style={font=\color{white!15!black},yshift=.2cm},
xlabel={Time ($10^{-1}$s)},
ymin=0,
ymax=9,
ylabel style={font=\color{white!15!black},yshift=-0.2cm},
ylabel={$(x_i)_1$},
axis background/.style={fill=white},
xmajorgrids,
ymajorgrids,
legend columns=2,
legend style={font=\tiny,at={(0.05,0.0)},opacity=0.6, anchor=south west, legend cell align=left, align=left, draw=white!15!black, text opacity=1}
]
\addplot [color=mycolor1, line width=2.0pt]
  table[row sep=crcr]{%
1	1\\
2	1\\
3	1\\
4	1.06476906552094\\
5	1.19431433040191\\
6	1.32476560557404\\
7	1.45642509779102\\
8	1.58950792849645\\
9	1.72495475709798\\
10	1.86165158569988\\
11	1.99834841430178\\
12	2.13504524290368\\
13	2.30321504366754\\
14	2.50285781659337\\
15	2.70250058951916\\
16	2.90214336244494\\
17	3.10178613537071\\
18	3.30142890829649\\
19	3.50107168122226\\
20	3.70071445414803\\
21	3.90035722707383\\
22	3.99999999999966\\
23	3.93208521270572\\
24	3.76096224742524\\
25	3.55418866437797\\
26	3.3474150813307\\
27	3.14064149828343\\
28	2.93386791523615\\
29	2.72709433218888\\
30	2.52032074914161\\
31	2.31354716609434\\
32	2.106773583047\\
33	1.99999999999957\\
34	2.06446439529923\\
35	2.23587791990558\\
36	2.44300259547152\\
37	2.65012727103757\\
38	2.85725194660371\\
39	3.06437662216985\\
40	3.27150129773599\\
41	3.47862597330213\\
42	3.68575064886827\\
43	3.89287532443457\\
44	4.00000000000102\\
45	3.9357054881783\\
46	3.76427765759487\\
47	3.5571356956399\\
48	3.34999373368494\\
49	3.14285177172997\\
50	2.935709809775\\
51	2.72856784782003\\
52	2.52142588586506\\
53	2.31428392391009\\
54	2.10714196195526\\
55	2.00000000000056\\
56	2.06428614744182\\
57	2.23571468239409\\
58	2.44285749546142\\
59	2.65000030852872\\
60	2.857143121596\\
61	3.06428593466329\\
62	3.27142874773057\\
63	3.47857156079785\\
64	3.68571437386513\\
65	3.89285718693231\\
66	3.99999999999938\\
67	3.93571426438667\\
68	3.7642856947541\\
69	3.55714283978144\\
70	3.34999998480878\\
71	3.14285712983612\\
72	2.93571427486345\\
73	2.72857141989079\\
74	2.52142856491813\\
75	2.31428570994547\\
76	2.10714285497264\\
77	1.99999999999963\\
78	2.0642857153353\\
79	2.23571428667549\\
80	2.44285714371152\\
81	2.65000000074753\\
82	2.85714285778352\\
83	3.0642857148195\\
84	3.27142857185549\\
85	3.47857142889147\\
86	3.68571428592746\\
87	3.89285714296411\\
88	4.00000000000143\\
89	3.93571428566461\\
90	3.76428571424019\\
91	3.55714285710234\\
92	3.34999999996451\\
93	3.14285714282668\\
94	2.93571428568885\\
95	2.72857142855101\\
96	2.52142857141318\\
97	2.31428571427535\\
98	2.10714285713759\\
99	1.99999999999989\\
100	2.06428571428812\\
101	2.23571428571652\\
102	2.44285714285912\\
103	2.65000000000171\\
104	2.85714285714429\\
105	3.06428571428687\\
106	3.27142857142945\\
107	3.47857142857203\\
108	3.68571428571461\\
109	3.89285714285713\\
110	3.99999999999957\\
111	3.93571428571368\\
112	3.76428571428512\\
113	3.55714285714225\\
114	3.34999999999937\\
115	3.14285714285649\\
116	2.93571428571361\\
117	2.72857142857074\\
118	2.52142857142786\\
119	2.31428571428498\\
120	2.10714285714217\\
121	1.99999999999942\\
122	2.0642857142852\\
123	2.23571428571381\\
124	2.44285714285672\\
125	2.64999999999961\\
126	2.8571428571425\\
127	3.06428571428539\\
128	3.27142857142828\\
129	3.47857142857117\\
130	3.68571428571406\\
131	3.89285714285688\\
};
\addlegendentry{Agent 1}

\addplot [color=mycolor2, line width=2.0pt]
  table[row sep=crcr]{%
1	5\\
2	4.9\\
3	4.66979166666673\\
4	4.40937500000019\\
5	4.14895833333366\\
6	3.88854166666712\\
7	3.62812500000058\\
8	3.36770833333404\\
9	3.1250000000006\\
10	3.00000000000026\\
11	3.07499999999993\\
12	3.34999999999959\\
13	3.74078947368363\\
14	4.14736842105206\\
15	4.55394736842048\\
16	4.96052631578891\\
17	5.36710526315733\\
18	5.77368421052576\\
19	6.18026315789419\\
20	6.58684210526261\\
21	6.89342105263104\\
22	6.99999999999947\\
23	6.90657894736789\\
24	6.61315789473632\\
25	6.21558171745105\\
26	5.81385041551208\\
27	5.41211911357311\\
28	5.01038781163414\\
29	4.60865650969517\\
30	4.20692520775619\\
31	3.80519390581722\\
32	3.40346260387825\\
33	3.10173130193928\\
34	3.00000000000031\\
35	3.09826869806134\\
36	3.39653739612237\\
37	3.79589954803956\\
38	4.19635515381291\\
39	4.59681075958626\\
40	4.99726636535961\\
41	5.39772197113296\\
42	5.79817757690631\\
43	6.19863318267966\\
44	6.59908878845309\\
45	6.8995443942266\\
46	7.00000000000011\\
47	6.90045560577362\\
48	6.60091121154713\\
49	6.20107906630577\\
50	5.80095917004953\\
51	5.4008392737933\\
52	5.00071937753706\\
53	4.60059948128082\\
54	4.20047958502459\\
55	3.80035968876835\\
56	3.40023979251211\\
57	3.10011989625588\\
58	2.99999999999964\\
59	3.0998801037434\\
60	3.39976020748717\\
61	3.79971603518224\\
62	4.19974758682863\\
63	4.59977913847502\\
64	4.9998106901214\\
65	5.39984224176779\\
66	5.79987379341418\\
67	6.19990534506056\\
68	6.59993689670703\\
69	6.89996844835358\\
70	7.00000000000013\\
71	6.90003155164667\\
72	6.60006310329322\\
73	6.20007472758409\\
74	5.80006642451927\\
75	5.40005812145445\\
76	5.00004981838963\\
77	4.60004151532481\\
78	4.20003321226\\
79	3.80002490919518\\
80	3.40001660613036\\
81	3.10000830306554\\
82	3.00000000000072\\
83	3.09999169693591\\
84	3.39998339387109\\
85	3.79997509080627\\
86	4.19997264872831\\
87	4.59997606763721\\
88	4.9999794865461\\
89	5.399982905455\\
90	5.7999863243639\\
91	6.19998974327279\\
92	6.59999316218177\\
93	6.89999658109083\\
94	6.99999999999989\\
95	6.90000341890895\\
96	6.600006837818\\
97	6.20001025672706\\
98	5.80001126228859\\
99	5.40000985450259\\
100	5.00000844671659\\
101	4.60000703893059\\
102	4.20000563114459\\
103	3.80000422335858\\
104	3.40000281557258\\
105	3.10000140778658\\
106	3.00000000000058\\
107	3.09999859221458\\
108	3.39999718442858\\
109	3.79999577664258\\
110	4.19999536258784\\
111	4.59999594226438\\
112	4.99999652194091\\
113	5.39999710161744\\
114	5.79999768129397\\
115	6.1999982609705\\
116	6.59999884064712\\
117	6.89999942032381\\
118	7.0000000000005\\
119	6.9000005796772\\
120	6.60000115935389\\
121	6.20000137291893\\
122	5.8000011885919\\
123	5.40000097248447\\
124	5.00000075637704\\
125	4.60000054026961\\
126	4.20000032416218\\
127	3.80000010805475\\
128	3.40000000000059\\
129	3.09999999999968\\
130	2.99999999999878\\
131	3.09999999999788\\
};
\addlegendentry{Agent 2}

\addplot [color=mycolor3, line width=2.0pt]
  table[row sep=crcr]{%
1	9\\
2	9\\
3	9\\
4	8.9352309344792\\
5	8.80568566959838\\
6	8.67523439442629\\
7	8.54357490220906\\
8	8.41049207150327\\
9	8.2750452429021\\
10	8.13834841430045\\
11	8.00165158569827\\
12	7.86495475709608\\
13	7.69678495633203\\
14	7.49714218340613\\
15	7.29749941048023\\
16	7.09785663755433\\
17	6.89821386462843\\
18	6.69857109170252\\
19	6.49892831877662\\
20	6.29928554585072\\
21	6.09964277292482\\
22	5.99999999999894\\
23	6.06791478729293\\
24	6.23903775257358\\
25	6.44581133562103\\
26	6.65258491866845\\
27	6.85935850171584\\
28	7.06613208476324\\
29	7.27290566781063\\
30	7.47967925085803\\
31	7.68645283390543\\
32	7.89322641695282\\
33	8.00000000000022\\
34	7.93553560470053\\
35	7.76412208009416\\
36	7.55699740452819\\
37	7.34987272896221\\
38	7.14274805339624\\
39	6.93562337783027\\
40	6.72849870226429\\
41	6.52137402669832\\
42	6.31424935113235\\
43	6.10712467556618\\
44	5.99999999999983\\
45	6.06429451182259\\
46	6.23572234240601\\
47	6.44286430436096\\
48	6.65000626631589\\
49	6.8571482282708\\
50	7.0642901902257\\
51	7.27143215218061\\
52	7.47857411413552\\
53	7.68571607609042\\
54	7.89285803804519\\
55	7.99999999999982\\
56	7.9357138525585\\
57	7.76428531760617\\
58	7.55714250453878\\
59	7.34999969147138\\
60	7.14285687840399\\
61	6.93571406533659\\
62	6.7285712522692\\
63	6.5214284392018\\
64	6.31428562613441\\
65	6.10714281306714\\
66	5.99999999999999\\
67	6.06428573561267\\
68	6.23571430524528\\
69	6.44285716021798\\
70	6.65000001519066\\
71	6.85714287016331\\
72	7.06428572513596\\
73	7.27142858010861\\
74	7.47857143508126\\
75	7.68571429005391\\
76	7.89285714502656\\
77	7.99999999999921\\
78	7.93571428466342\\
79	7.76428571332334\\
80	7.55714285628742\\
81	7.3499999992515\\
82	7.14285714221557\\
83	6.93571428517965\\
84	6.72857142814373\\
85	6.5214285711078\\
86	6.31428571407188\\
87	6.10714285703567\\
88	5.99999999999917\\
89	6.06428571433639\\
90	6.23571428576071\\
91	6.44285714289837\\
92	6.6500000000361\\
93	6.85714285717393\\
94	7.06428571431176\\
95	7.27142857144959\\
96	7.47857142858741\\
97	7.68571428572524\\
98	7.89285714286301\\
99	8.00000000000071\\
100	7.93571428571247\\
101	7.76428571428407\\
102	7.55714285714146\\
103	7.34999999999888\\
104	7.1428571428563\\
105	6.93571428571372\\
106	6.72857142857117\\
107	6.52142857142861\\
108	6.31428571428605\\
109	6.10714285714319\\
110	6.00000000000003\\
111	6.06428571428561\\
112	6.23571428571418\\
113	6.44285714285706\\
114	6.65000000000002\\
115	6.85714285714298\\
116	7.06428571428594\\
117	7.2714285714289\\
118	7.47857142857187\\
119	7.68571428571483\\
120	7.89285714285772\\
121	8.00000000000055\\
122	7.93571428571481\\
123	7.76428571428616\\
124	7.55714285714322\\
125	7.35000000000029\\
126	7.14285714285737\\
127	6.93571428571445\\
128	6.72857142857152\\
129	6.5214285714286\\
130	6.31428571428568\\
131	6.10714285714282\\
};
\addlegendentry{Agent 3}

\end{axis}

\end{tikzpicture}
    \label{fig:xk1}
  \end{subfigure}
  \begin{subfigure}{0.245\textwidth} 
    \definecolor{mycolor1}{rgb}{0.00000,0.44700,0.74100}%
\definecolor{mycolor2}{rgb}{0.85000,0.32500,0.09800}%
\definecolor{mycolor3}{rgb}{0.92900,0.69400,0.12500}%
\begin{tikzpicture}

\begin{axis}[%
width=1.3in,
height=1.16in,
at={(0.in,0.in)},
scale only axis,
xmin=0,
xmax=130,
xlabel style={font=\color{white!15!black},yshift=.2cm},
xlabel={Time ($10^{-1}$s)},
ymin=0,
ymax=9,
ylabel style={font=\color{white!15!black},yshift=-0.2cm},
ylabel={$(x_i)_3$},
axis background/.style={fill=white},
xmajorgrids,
ymajorgrids,
legend columns=2,
legend style={font=\tiny,at={(0.05,0.0)},opacity=0.6, anchor=south west, legend cell align=left, align=left, draw=white!15!black, text opacity=1}
]
\addplot [color=mycolor1, line width=2.0pt]
  table[row sep=crcr]{%
1	9\\
2	8.9\\
3	8.62941176470584\\
4	8.28823529411752\\
5	7.9470588235292\\
6	7.60588235294087\\
7	7.26470588235254\\
8	6.92352941176422\\
9	6.58235294117589\\
10	6.24117647058756\\
11	5.99999999999923\\
12	5.9588235294109\\
13	6.09938080495272\\
14	6.32167182662469\\
15	6.54396284829664\\
16	6.76625386996856\\
17	6.98854489164049\\
18	7.21083591331241\\
19	7.43312693498433\\
20	7.65541795665626\\
21	7.87770897832837\\
22	8.00000000000067\\
23	7.9430440427453\\
24	7.77099822861939\\
25	7.56310953655062\\
26	7.35522084448184\\
27	7.14733215241307\\
28	6.93944346034429\\
29	6.73155476827552\\
30	6.52366607620674\\
31	6.31577738413797\\
32	6.10788869206913\\
33	6.00000000000024\\
34	6.0639248264185\\
35	6.23538378840431\\
36	6.44256336747054\\
37	6.64974294653674\\
38	6.85692252560291\\
39	7.06410210466909\\
40	7.27128168373526\\
41	7.47846126280143\\
42	7.68564084186761\\
43	7.89282042093407\\
44	8.00000000000082\\
45	7.93573205438792\\
46	7.76430198665\\
47	7.55715732146673\\
48	7.35001265628346\\
49	7.14286799110019\\
50	6.93572332591692\\
51	6.72857866073364\\
52	6.52143399555037\\
53	6.3142893303671\\
54	6.10714466518302\\
55	5.99999999999812\\
56	6.06428483942486\\
57	6.2357134845259\\
58	6.44285643068961\\
59	6.64999937685342\\
60	6.85714232301732\\
61	7.06428526918123\\
62	7.27142821534513\\
63	7.47857116150904\\
64	7.68571410767294\\
65	7.89285705383702\\
66	8.00000000000127\\
67	7.93571432879037\\
68	7.76428575373437\\
69	7.55714289220839\\
70	7.35000003068242\\
71	7.14285716915644\\
72	6.93571430763047\\
73	6.72857144610449\\
74	6.52142858457852\\
75	6.31428572305254\\
76	6.1071428615258\\
77	5.9999999999983\\
78	6.0642857121625\\
79	6.23571428376987\\
80	6.44285714112871\\
81	6.64999999848765\\
82	6.85714285584668\\
83	7.06428571320572\\
84	7.27142857056475\\
85	7.47857142792379\\
86	7.68571428528282\\
87	7.89285714264188\\
88	8.00000000000097\\
89	7.93571428581986\\
90	7.76428571438511\\
91	7.55714285723691\\
92	7.35000000008874\\
93	7.14285714294056\\
94	6.93571428579239\\
95	6.72857142864422\\
96	6.52142857149604\\
97	6.31428571434787\\
98	6.10714285719976\\
99	6.00000000005172\\
100	6.06428571432961\\
101	6.23571428574767\\
102	6.44285714287993\\
103	6.65000000001217\\
104	6.85714285714441\\
105	7.06428571427664\\
106	7.27142857140888\\
107	7.47857142854112\\
108	7.68571428567336\\
109	7.89285714280553\\
110	7.99999999993764\\
111	7.9357142856414\\
112	7.7642857142025\\
113	7.55714285704928\\
114	7.34999999989606\\
115	7.14285714274284\\
116	6.93571428558962\\
117	6.7285714284364\\
118	6.52142857128318\\
119	6.31428571412996\\
120	6.10714285697681\\
121	5.99999999982372\\
122	6.06428571409915\\
123	6.23571428551742\\
124	6.44285714264999\\
125	6.64999999978254\\
126	6.85714285691509\\
127	7.06428571404764\\
128	7.27142857118019\\
129	7.47857142831273\\
130	7.68571428544528\\
131	7.89285714257776\\
};
\addlegendentry{Agent 1}

\addplot [color=mycolor2, line width=2.0pt]
  table[row sep=crcr]{%
1	5\\
2	4.9\\
3	4.66979166666673\\
4	4.40937500000019\\
5	4.14895833333366\\
6	3.88854166666712\\
7	3.62812500000058\\
8	3.36770833333404\\
9	3.1250000000006\\
10	3.00000000000026\\
11	3.07499999999993\\
12	3.34999999999959\\
13	3.74078947368363\\
14	4.14736842105206\\
15	4.55394736842048\\
16	4.96052631578891\\
17	5.36710526315733\\
18	5.77368421052576\\
19	6.18026315789419\\
20	6.58684210526261\\
21	6.89342105263104\\
22	6.99999999999947\\
23	6.90657894736789\\
24	6.61315789473632\\
25	6.21558171745105\\
26	5.81385041551208\\
27	5.41211911357311\\
28	5.01038781163414\\
29	4.60865650969517\\
30	4.20692520775619\\
31	3.80519390581722\\
32	3.40346260387825\\
33	3.10173130193928\\
34	3.00000000000031\\
35	3.09826869806134\\
36	3.39653739612237\\
37	3.79589954803956\\
38	4.19635515381291\\
39	4.59681075958626\\
40	4.99726636535961\\
41	5.39772197113296\\
42	5.79817757690631\\
43	6.19863318267966\\
44	6.59908878845309\\
45	6.8995443942266\\
46	7.00000000000011\\
47	6.90045560577362\\
48	6.60091121154713\\
49	6.20107906630577\\
50	5.80095917004953\\
51	5.4008392737933\\
52	5.00071937753706\\
53	4.60059948128082\\
54	4.20047958502459\\
55	3.80035968876835\\
56	3.40023979251211\\
57	3.10011989625588\\
58	2.99999999999964\\
59	3.0998801037434\\
60	3.39976020748717\\
61	3.79971603518224\\
62	4.19974758682863\\
63	4.59977913847502\\
64	4.9998106901214\\
65	5.39984224176779\\
66	5.79987379341418\\
67	6.19990534506056\\
68	6.59993689670703\\
69	6.89996844835358\\
70	7.00000000000013\\
71	6.90003155164667\\
72	6.60006310329322\\
73	6.20007472758409\\
74	5.80006642451927\\
75	5.40005812145445\\
76	5.00004981838963\\
77	4.60004151532481\\
78	4.20003321226\\
79	3.80002490919518\\
80	3.40001660613036\\
81	3.10000830306554\\
82	3.00000000000072\\
83	3.09999169693591\\
84	3.39998339387109\\
85	3.79997509080627\\
86	4.19997264872831\\
87	4.59997606763721\\
88	4.9999794865461\\
89	5.399982905455\\
90	5.7999863243639\\
91	6.19998974327279\\
92	6.59999316218177\\
93	6.89999658109083\\
94	6.99999999999989\\
95	6.90000341890895\\
96	6.600006837818\\
97	6.20001025672706\\
98	5.80001126228859\\
99	5.40000985450259\\
100	5.00000844671659\\
101	4.60000703893059\\
102	4.20000563114459\\
103	3.80000422335858\\
104	3.40000281557258\\
105	3.10000140778658\\
106	3.00000000000058\\
107	3.09999859221458\\
108	3.39999718442858\\
109	3.79999577664258\\
110	4.19999536258784\\
111	4.59999594226438\\
112	4.99999652194091\\
113	5.39999710161744\\
114	5.79999768129397\\
115	6.1999982609705\\
116	6.59999884064712\\
117	6.89999942032381\\
118	7.0000000000005\\
119	6.9000005796772\\
120	6.60000115935389\\
121	6.20000137291893\\
122	5.8000011885919\\
123	5.40000097248447\\
124	5.00000075637704\\
125	4.60000054026961\\
126	4.20000032416218\\
127	3.80000010805475\\
128	3.40000000000059\\
129	3.09999999999968\\
130	2.99999999999878\\
131	3.09999999999788\\
};
\addlegendentry{Agent 2}

\addplot [color=mycolor3, line width=2.0pt]
  table[row sep=crcr]{%
1	1\\
2	1.1\\
3	1.37058823529416\\
4	1.71176470588248\\
5	2.0529411764708\\
6	2.39411764705912\\
7	2.73529411764744\\
8	3.07647058823577\\
9	3.41764705882409\\
10	3.75882352941241\\
11	4.00000000000073\\
12	4.04117647058905\\
13	3.90061919504724\\
14	3.67832817337531\\
15	3.45603715170337\\
16	3.23374613003143\\
17	3.01145510835949\\
18	2.78916408668756\\
19	2.56687306501562\\
20	2.34458204334368\\
21	2.12229102167156\\
22	1.99999999999924\\
23	2.0569559572546\\
24	2.22900177138049\\
25	2.43689046344925\\
26	2.64477915551811\\
27	2.85266784758707\\
28	3.06055653965602\\
29	3.26844523172498\\
30	3.47633392379393\\
31	3.68422261586289\\
32	3.89211130793184\\
33	4.0000000000008\\
34	3.93607517358249\\
35	3.7646162115966\\
36	3.55743663253039\\
37	3.35025705346419\\
38	3.14307747439798\\
39	2.93589789533178\\
40	2.72871831626557\\
41	2.52153873719937\\
42	2.31435915813317\\
43	2.10717957906628\\
44	1.99999999999871\\
45	2.06426794561121\\
46	2.23569801334923\\
47	2.4428426785327\\
48	2.64998734371615\\
49	2.85713200889957\\
50	3.06427667408298\\
51	3.2714213392664\\
52	3.47856600444982\\
53	3.68571066963324\\
54	3.89285533481645\\
55	3.99999999999943\\
56	3.93571516057183\\
57	3.76428651547098\\
58	3.55714356930746\\
59	3.35000062314395\\
60	3.14285767698044\\
61	2.93571473081692\\
62	2.72857178465341\\
63	2.52142883848989\\
64	2.31428589232638\\
65	2.10714294616263\\
66	1.99999999999864\\
67	2.06428567120967\\
68	2.2357142462657\\
69	2.4428571077917\\
70	2.64999996931768\\
71	2.85714283084363\\
72	3.06428569236958\\
73	3.27142855389553\\
74	3.47857141542148\\
75	3.68571427694743\\
76	3.89285713847338\\
77	3.99999999999934\\
78	3.93571428783446\\
79	3.7642857162273\\
80	3.55714285886867\\
81	3.35000000151004\\
82	3.1428571441514\\
83	2.93571428679277\\
84	2.72857142943414\\
85	2.52142857207551\\
86	2.31428571471688\\
87	2.10714285735816\\
88	1.99999999999934\\
89	2.0642857141806\\
90	2.2357142856153\\
91	2.44285714276334\\
92	2.64999999991146\\
93	2.85714285705967\\
94	3.06428571420788\\
95	3.2714285713561\\
96	3.47857142850431\\
97	3.68571428565252\\
98	3.89285714280067\\
99	3.99999999994875\\
100	3.9357142856709\\
101	3.76428571425288\\
102	3.55714285712066\\
103	3.34999999998846\\
104	3.14285714285626\\
105	2.93571428572408\\
106	2.7285714285919\\
107	2.52142857145973\\
108	2.31428571432756\\
109	2.10714285719508\\
110	2.0000000000623\\
111	2.06428571435827\\
112	2.23571428579722\\
113	2.44285714295048\\
114	2.65000000010383\\
115	2.85714285725717\\
116	3.06428571441052\\
117	3.27142857156387\\
118	3.47857142871721\\
119	3.68571428587056\\
120	3.89285714302384\\
121	4.00000000017705\\
122	3.93571428590169\\
123	3.76428571448343\\
124	3.55714285735087\\
125	3.35000000021833\\
126	3.14285714308579\\
127	2.93571428595325\\
128	2.72857142882071\\
129	2.52142857168818\\
130	2.31428571455564\\
131	2.10714285742316\\
};
\addlegendentry{Agent 3}

\end{axis}

\end{tikzpicture}%
    \label{fig:xk3}
  \end{subfigure}
  \hfill 
  \begin{subfigure}{0.245\textwidth} 
   \definecolor{mycolor1}{rgb}{0.00000,0.44700,0.74100}%
\definecolor{mycolor2}{rgb}{0.85000,0.32500,0.09800}%
\definecolor{mycolor3}{rgb}{0.92900,0.69400,0.12500}%
\begin{tikzpicture}

\begin{axis}[%
width=1.3in,
height=1.16in,
at={(0.in,0.in)},
scale only axis,
xmin=0,
xmax=130,
xlabel style={font=\color{white!15!black},yshift=.2cm},
xlabel={Time ($10^{-1}$s)},
ymin=-20.5,
ymax=20.5,
ylabel style={font=\color{white!15!black},yshift=-0.5cm},
ylabel={$(u_i)_1$},
axis background/.style={fill=white},
xmajorgrids,
ymajorgrids,
legend columns=2,
legend style={font=\tiny,at={(0.05,0.0)},opacity=0.6, anchor=south west, legend cell align=left, align=left, draw=white!15!black, text opacity=1}
]
\addplot [color=mycolor1, line width=1.0pt]
  table[row sep=crcr]{%
1	0\\
2	0\\
3	12.9538131041888\\
4	0.0014267678151656\\
5	0.17977529041901\\
6	0.0618681185493445\\
7	0.222799579142659\\
8	0.250000000075186\\
9	0\\
10	0\\
11	0\\
12	6.29459443239212\\
13	0\\
14	-4.84590145788388e-12\\
15	0\\
16	0\\
17	0\\
18	1.77635683940025e-15\\
19	0\\
20	5.24499447745551e-12\\
21	-20\\
22	-13.5115120439525\\
23	-7.13012355335721\\
24	0\\
25	0\\
26	0\\
27	0\\
28	0\\
29	0\\
30	0\\
31	-1.50803515944404e-11\\
32	20\\
33	14.2475956694163\\
34	7.14223019192056\\
35	0\\
36	1.94049221136083e-11\\
37	0\\
38	0\\
39	0\\
40	0\\
41	0\\
42	3.1467971052527e-11\\
43	-20\\
44	-14.2838374778348\\
45	-7.14282627430772\\
46	0\\
47	0\\
48	0\\
49	0\\
50	0\\
51	0\\
52	0\\
53	2.75317105862818e-11\\
54	20\\
55	14.28562187919\\
56	7.14285562301244\\
57	0\\
58	-4.84945417156268e-12\\
59	0\\
60	0\\
61	0\\
62	0\\
63	0\\
64	-2.09739904732039e-11\\
65	-20\\
66	-14.2857097359557\\
67	-7.14285706801762\\
68	0\\
69	0\\
70	0\\
71	0\\
72	0\\
73	1.77635683940025e-15\\
74	0\\
75	-3.47420759369108e-11\\
76	20\\
77	14.2857140617355\\
78	7.14285713916883\\
79	0\\
80	-4.84945417156268e-12\\
81	0\\
82	0\\
83	0\\
84	0\\
85	0\\
86	1.33056808152985e-10\\
87	-20\\
88	-14.2857142748272\\
89	-7.14285714269193\\
90	4.2774672692758e-12\\
91	0\\
92	0\\
93	0\\
94	0\\
95	-1.77635683940025e-15\\
96	0\\
97	1.32311939182728e-11\\
98	20\\
99	14.2857142851872\\
100	7.14285714284511\\
101	-4.2810199829546e-12\\
102	0\\
103	0\\
104	0\\
105	0\\
106	0\\
107	0\\
108	-1.32391875240501e-11\\
109	-20\\
110	-14.2857142856683\\
111	-7.14285714286425\\
112	0\\
113	0\\
114	0\\
115	0\\
116	0\\
117	1.77635683940025e-15\\
118	0\\
119	1.32329702751122e-11\\
120	20\\
121	14.2857142857036\\
122	7.14285714286425\\
123	-4.2774672692758e-12\\
124	0\\
125	0\\
126	0\\
127	0\\
128	0\\
129	0\\
130	-1.32418520593092e-11\\
};
\addlegendentry{Agent 1}

\addplot [color=mycolor2, line width=1.0pt]
  table[row sep=crcr]{%
1	-20\\
2	-6.04166666665378\\
3	0\\
4	0\\
5	0\\
6	0\\
7	0\\
8	3.54166666661984\\
9	20\\
10	20\\
11	20\\
12	3.15789473687657\\
13	0\\
14	0\\
15	0\\
16	0\\
17	0\\
18	8.88178419700125e-16\\
19	0\\
20	-20\\
21	-20\\
22	-20\\
23	-20\\
24	-0.831024930739659\\
25	0\\
26	0\\
27	0\\
28	0\\
29	0\\
30	0\\
31	0\\
32	20\\
33	20\\
34	20\\
35	20\\
36	0.218690771231962\\
37	0\\
38	0\\
39	0\\
40	0\\
41	0\\
42	0\\
43	1.61151092470391e-11\\
44	-20\\
45	-20\\
46	-20\\
47	-20\\
48	-0.0575502029747241\\
49	0\\
50	0\\
51	0\\
52	0\\
53	0\\
54	0\\
55	0\\
56	20\\
57	20\\
58	20\\
59	20\\
60	0.0151447902623509\\
61	0\\
62	0\\
63	0\\
64	0\\
65	0\\
66	0\\
67	1.61151092470391e-11\\
68	-20\\
69	-20\\
70	-20\\
71	-20\\
72	-0.00398547113658428\\
73	0\\
74	0\\
75	0\\
76	0\\
77	0\\
78	0\\
79	0\\
80	20\\
81	20\\
82	20\\
83	20\\
84	0\\
85	0.00117219737148491\\
86	0\\
87	0\\
88	0\\
89	0\\
90	0\\
91	1.61151092470391e-11\\
92	-20\\
93	-20\\
94	-20\\
95	-20\\
96	0\\
97	-0.000482669505906586\\
98	0\\
99	0\\
100	0\\
101	0\\
102	0\\
103	0\\
104	20\\
105	20\\
106	20\\
107	20\\
108	0\\
109	0.000198746253287462\\
110	0\\
111	0\\
112	0\\
113	0\\
114	0\\
115	1.61151092470391e-11\\
116	-20\\
117	-20\\
118	-20\\
119	-20\\
120	-7.32223308759217e-05\\
121	-6.35608144325815e-06\\
122	0\\
123	0\\
124	0\\
125	0\\
126	0\\
127	2.16106526949034e-05\\
128	20\\
129	20\\
130	20\\
};
\addlegendentry{Agent 2}

\addplot [color=mycolor3, line width=1.0pt]
  table[row sep=crcr]{%
1	0\\
2	0\\
3	-12.9538131041603\\
4	-0.00142676784388795\\
5	-0.179775290408951\\
6	-0.0618681186194472\\
7	-0.222799579092854\\
8	-0.249999999984237\\
9	-1.09139364212751e-10\\
10	0\\
11	0\\
12	-6.2945944323713\\
13	0\\
14	0\\
15	0\\
16	0\\
17	0\\
18	-1.77635683940025e-15\\
19	0\\
20	1.30925233535639e-12\\
21	20\\
22	13.511512043976\\
23	7.13012355335721\\
24	0\\
25	-4.84590145788388e-12\\
26	0\\
27	0\\
28	0\\
29	0\\
30	0\\
31	0\\
32	-20\\
33	-14.2475956694163\\
34	-7.14223019192056\\
35	0\\
36	0\\
37	0\\
38	0\\
39	0\\
40	0\\
41	0\\
42	-3.80186982646099e-11\\
43	20\\
44	14.283837477823\\
45	7.14282627430772\\
46	0\\
47	-4.84945417156268e-12\\
48	0\\
49	1.0842021724855e-19\\
50	0\\
51	0\\
52	0\\
53	-2.75282377578173e-11\\
54	-20\\
55	-14.28562187919\\
56	-7.14285562301244\\
57	0\\
58	0\\
59	0\\
60	0\\
61	0\\
62	0\\
63	0\\
64	2.42550969993692e-11\\
65	20\\
66	14.2857097359675\\
67	7.14285706801763\\
68	0\\
69	-4.84590145788388e-12\\
70	0\\
71	0\\
72	0\\
73	0\\
74	0\\
75	0\\
76	-20\\
77	-14.2857140616886\\
78	-7.14285713916883\\
79	0\\
80	0\\
81	0\\
82	0\\
83	0\\
84	0\\
85	0\\
86	-5.76829435921805e-11\\
87	20\\
88	14.285714274745\\
89	7.14285714267278\\
90	-4.2774672692758e-12\\
91	1.94013693999295e-11\\
92	0\\
93	0\\
94	0\\
95	0\\
96	0\\
97	-1.32360788995811e-11\\
98	-20\\
99	-14.2857142851872\\
100	-7.14285714284511\\
101	4.2810199829546e-12\\
102	0\\
103	0\\
104	2.59801451942922e-12\\
105	0\\
106	0\\
107	0\\
108	-6.07760508586352e-11\\
109	20\\
110	14.2857142857505\\
111	7.14285714284511\\
112	1.71205272181396e-11\\
113	0\\
114	0\\
115	0\\
116	-1.77635683940025e-15\\
117	0\\
118	0\\
119	-1.32360788995811e-11\\
120	-20\\
121	-14.2857142857153\\
122	-7.14285714286425\\
123	4.2810199829546e-12\\
124	0\\
125	0\\
126	0\\
127	0\\
128	1.77635683940025e-15\\
129	0\\
130	1.32427402377289e-11\\
};
\addlegendentry{Agent 3}

\end{axis}

\end{tikzpicture}%
    \label{fig:uk1}
  \end{subfigure}
  \begin{subfigure}{0.245\textwidth} 
    \definecolor{mycolor1}{rgb}{0.00000,0.44700,0.74100}%
\definecolor{mycolor2}{rgb}{0.85000,0.32500,0.09800}%
\definecolor{mycolor3}{rgb}{0.92900,0.69400,0.12500}%
\begin{tikzpicture}

\begin{axis}[%
width=1.3in,
height=1.16in,
at={(0.in,0.in)},
scale only axis,
xmin=0,
xmax=130,
xlabel style={font=\color{white!15!black},yshift=.2cm},
xlabel={Time ($10^{-1}$s)},
ymin=-20.5,
ymax=20.5,
ylabel style={font=\color{white!15!black},yshift=-0.5cm},
ylabel={$(u_i)_2$},
axis background/.style={fill=white},
xmajorgrids,
ymajorgrids,
legend columns=2,
legend style={font=\tiny,at={(0.05,0.0)},opacity=0.6, anchor=south west, legend cell align=left, align=left, draw=white!15!black, text opacity=1}
]
\addplot [color=mycolor1, line width=1.0pt]
  table[row sep=crcr]{%
1	-20\\
2	-14.1176470588321\\
3	0\\
4	0\\
5	-7.86037901434611e-13\\
6	0\\
7	0\\
8	0\\
9	0\\
10	20\\
11	20\\
12	16.34674922603\\
13	0\\
14	-4.85048112786046e-12\\
15	0\\
16	0\\
17	0\\
18	0\\
19	0\\
20	3.80172014065342e-11\\
21	-20\\
22	-15.8493957855347\\
23	-7.16857558857316\\
24	0\\
25	0\\
26	0\\
27	0\\
28	0\\
29	0\\
30	0\\
31	-1.18012290687895e-11\\
32	20\\
33	14.3627036974302\\
34	7.14412341608151\\
35	0\\
36	-4.84945417156268e-12\\
37	0\\
38	0\\
39	0\\
40	0\\
41	0\\
42	5.76844196631077e-11\\
43	-20\\
44	-14.2895049359316\\
45	-7.14291948907071\\
46	0\\
47	0\\
48	0\\
49	0\\
50	0\\
51	0\\
52	0\\
53	-1.62559086986735e-10\\
54	20\\
55	14.2859009223274\\
56	7.14286021253333\\
57	0\\
58	1.94049221136083e-11\\
59	0\\
60	0\\
61	0\\
62	0\\
63	0\\
64	3.47435727949865e-11\\
65	-20\\
66	-14.2857234750291\\
67	-7.14285729399355\\
68	0\\
69	0\\
70	0\\
71	-5.55111512312578e-17\\
72	0\\
73	-1.77635683940025e-15\\
74	0\\
75	-1.52731441030921e-10\\
76	20\\
77	14.2857147383409\\
78	7.14285715029339\\
79	0\\
80	1.94049221136083e-11\\
81	0\\
82	0\\
83	0\\
84	-1.77635683940025e-15\\
85	0\\
86	5.24648755531548e-12\\
87	-20\\
88	-14.2857143080384\\
89	-7.14285714269193\\
90	4.2774672692758e-12\\
91	0\\
92	0\\
93	0\\
94	0\\
95	-1.77635683940025e-15\\
96	0\\
97	1.32311939182728e-11\\
98	20\\
99	14.2857142851872\\
100	7.14285714284511\\
101	-4.2810199829546e-12\\
102	0\\
103	0\\
104	0\\
105	0\\
106	0\\
107	0\\
108	-1.32391875240501e-11\\
109	-20\\
110	-14.2857142856683\\
111	-7.14285714286425\\
112	0\\
113	0\\
114	0\\
115	0\\
116	0\\
117	1.77635683940025e-15\\
118	0\\
119	1.32329702751122e-11\\
120	20\\
121	14.2857142857036\\
122	7.14285714286425\\
123	-4.2774672692758e-12\\
124	0\\
125	0\\
126	0\\
127	0\\
128	0\\
129	0\\
130	-1.32418520593092e-11\\
};
\addlegendentry{Agent 1}

\addplot [color=mycolor2, line width=1.0pt]
  table[row sep=crcr]{%
1	-20\\
2	-6.04166666665378\\
3	0\\
4	0\\
5	0\\
6	0\\
7	0\\
8	3.54166666661984\\
9	20\\
10	20\\
11	20\\
12	3.15789473687657\\
13	0\\
14	0\\
15	0\\
16	0\\
17	0\\
18	8.88178419700125e-16\\
19	0\\
20	-20\\
21	-20\\
22	-20\\
23	-20\\
24	-0.831024930739659\\
25	0\\
26	0\\
27	0\\
28	0\\
29	0\\
30	0\\
31	0\\
32	20\\
33	20\\
34	20\\
35	20\\
36	0.218690771231962\\
37	0\\
38	0\\
39	0\\
40	0\\
41	0\\
42	0\\
43	1.61151092470391e-11\\
44	-20\\
45	-20\\
46	-20\\
47	-20\\
48	-0.0575502029747241\\
49	0\\
50	0\\
51	0\\
52	0\\
53	0\\
54	0\\
55	0\\
56	20\\
57	20\\
58	20\\
59	20\\
60	0.0151447902623509\\
61	0\\
62	0\\
63	0\\
64	0\\
65	0\\
66	0\\
67	1.61151092470391e-11\\
68	-20\\
69	-20\\
70	-20\\
71	-20\\
72	-0.00398547113658428\\
73	0\\
74	0\\
75	0\\
76	0\\
77	0\\
78	0\\
79	0\\
80	20\\
81	20\\
82	20\\
83	20\\
84	0\\
85	0.00117219737148491\\
86	0\\
87	0\\
88	0\\
89	0\\
90	0\\
91	1.61151092470391e-11\\
92	-20\\
93	-20\\
94	-20\\
95	-20\\
96	0\\
97	-0.000482669505906586\\
98	0\\
99	0\\
100	0\\
101	0\\
102	0\\
103	0\\
104	20\\
105	20\\
106	20\\
107	20\\
108	0\\
109	0.000198746253287462\\
110	0\\
111	0\\
112	0\\
113	0\\
114	0\\
115	1.61151092470391e-11\\
116	-20\\
117	-20\\
118	-20\\
119	-20\\
120	-7.32223308759217e-05\\
121	-6.35608144325815e-06\\
122	0\\
123	0\\
124	0\\
125	0\\
126	0\\
127	2.16106526949034e-05\\
128	20\\
129	20\\
130	20\\
};
\addlegendentry{Agent 2}

\addplot [color=mycolor3, line width=1.0pt]
  table[row sep=crcr]{%
1	20\\
2	14.1176470588321\\
3	0\\
4	0\\
5	0\\
6	0\\
7	0\\
8	0\\
9	0\\
10	-20\\
11	-20\\
12	-16.3467492260258\\
13	0\\
14	0\\
15	0\\
16	0\\
17	0\\
18	0\\
19	0\\
20	-3.80172014065342e-11\\
21	20\\
22	15.8493957855347\\
23	7.16857558857315\\
24	0\\
25	1.94049221136083e-11\\
26	0\\
27	0\\
28	0\\
29	0\\
30	0\\
31	0\\
32	-20\\
33	-14.3627036974536\\
34	-7.14412341606236\\
35	0\\
36	0\\
37	0\\
38	0\\
39	0\\
40	1.77635683940025e-15\\
41	0\\
42	-1.36337725943056e-10\\
43	20\\
44	14.2895049360138\\
45	7.14291948908986\\
46	0\\
47	-4.84945417156268e-12\\
48	0\\
49	1.0842021724855e-19\\
50	0\\
51	0\\
52	0\\
53	-4.39163190828533e-11\\
54	-20\\
55	-14.2859009221161\\
56	-7.14286021253333\\
57	0\\
58	0\\
59	0\\
60	0\\
61	0\\
62	0\\
63	0\\
64	-4.78530558219345e-11\\
65	20\\
66	14.2857234750057\\
67	7.14285729399355\\
68	0\\
69	-4.84945417156268e-12\\
70	0\\
71	0\\
72	0\\
73	1.77635683940025e-15\\
74	0\\
75	0\\
76	-20\\
77	-14.2857147381648\\
78	-7.14285715029339\\
79	0\\
80	0\\
81	0\\
82	0\\
83	0\\
84	1.77635683940025e-15\\
85	0\\
86	-1.83574674403392e-11\\
87	20\\
88	14.2857143080149\\
89	7.14285714267278\\
90	-4.2774672692758e-12\\
91	1.94013693999295e-11\\
92	0\\
93	0\\
94	0\\
95	0\\
96	0\\
97	-1.32360788995811e-11\\
98	-20\\
99	-14.2857142851872\\
100	-7.14285714284511\\
101	4.2810199829546e-12\\
102	0\\
103	0\\
104	2.59801451942922e-12\\
105	0\\
106	0\\
107	0\\
108	-6.07760508586352e-11\\
109	20\\
110	14.2857142857505\\
111	7.14285714284511\\
112	1.71205272181396e-11\\
113	0\\
114	0\\
115	0\\
116	-1.77635683940025e-15\\
117	0\\
118	0\\
119	-1.32360788995811e-11\\
120	-20\\
121	-14.2857142857153\\
122	-7.14285714286425\\
123	4.2810199829546e-12\\
124	0\\
125	0\\
126	0\\
127	0\\
128	1.77635683940025e-15\\
129	0\\
130	1.32427402377289e-11\\
};
\addlegendentry{Agent 3}

\end{axis}

\end{tikzpicture}%
    \label{fig:uk2}
  \end{subfigure}
  \caption{Position and acceleration over a 13-second time horizon under the distributed receding horizon control scheme.}
  \label{fig:main}
\end{figure*}
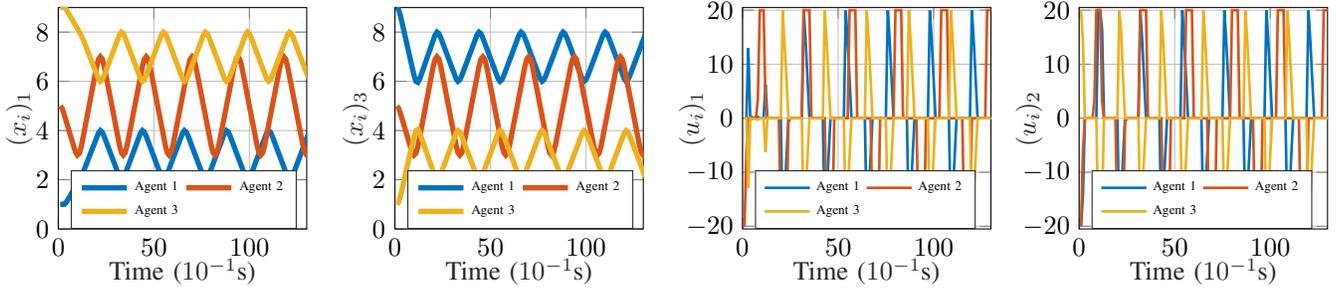

\begin{figure*}
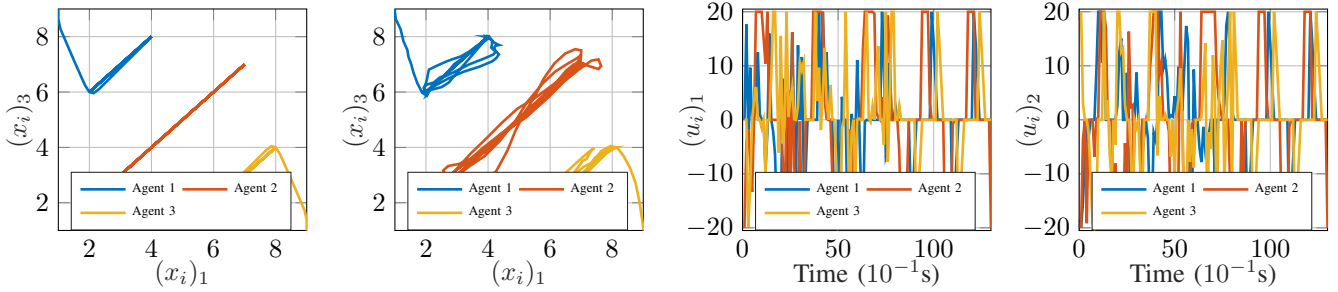
 
  \centering
  \begin{subfigure}{0.245\textwidth} 
    \input{figures/norm_traj}
    \label{fig:nom_traj}
  \end{subfigure}
  \begin{subfigure}{0.245\textwidth} 
    \input{figures/pert_traj}
    \label{fig:per_traj}
  \end{subfigure}
  \hfill 
  \begin{subfigure}{0.245\textwidth} 
   \definecolor{mycolor1}{rgb}{0.00000,0.44700,0.74100}%
\definecolor{mycolor2}{rgb}{0.85000,0.32500,0.09800}%
\definecolor{mycolor3}{rgb}{0.92900,0.69400,0.12500}%
\begin{tikzpicture}

\begin{axis}[%
width=1.3in,
height=1.16in,
at={(0.in,0.in)},
scale only axis,
xmin=0,
xmax=130,
xlabel style={font=\color{white!15!black},yshift=.2cm},
xlabel={Time ($10^{-1}$s)},
ymin=-20.5,
ymax=20.5,
ylabel style={font=\color{white!15!black},yshift=-0.5cm},
ylabel={$(u_i)_1$},
axis background/.style={fill=white},
xmajorgrids,
ymajorgrids,
legend columns=2,
legend style={font=\tiny,at={(0.05,0.0)},opacity=0.6, anchor=south west, legend cell align=left, align=left, draw=white!15!black, text opacity=1}
]
\addplot [color=mycolor1, line width=1.0pt]
  table[row sep=crcr]{%
1	0\\
2	17.6507429123026\\
3	-3.60613101547642\\
4	9.65888120499085\\
5	-1.01626848428912\\
6	0\\
7	4.99180266867242\\
8	12.4983137486793\\
9	1.08362853072928\\
10	0\\
11	0\\
12	-9.16569538797679\\
13	-2.24070438913865\\
14	15.9769848593372\\
15	0\\
16	-2.13954635266494\\
17	-6.44366097175104\\
18	-18.5297500106329\\
19	-20\\
20	-20\\
21	7.59667804548891\\
22	3.88442822617695\\
23	0\\
24	-20\\
25	-4.44089209850063e-16\\
26	16.039023334694\\
27	-2.19317182716622\\
28	-9.86992749498195\\
29	13.8144131354056\\
30	0\\
31	11.6374549555621\\
32	-9.91470549459336\\
33	-10.0203111150821\\
34	0\\
35	0\\
36	0\\
37	0\\
38	0\\
39	0\\
40	20\\
41	0.77048408237772\\
42	20\\
43	7.75766230859777\\
44	12.4699992401675\\
45	0\\
46	0\\
47	0\\
48	-20\\
49	-15.4988098419433\\
50	0\\
51	-20\\
52	4.25993421195017\\
53	0\\
54	0\\
55	-20\\
56	-13.8327000460968\\
57	1.10430411114066\\
58	-1.4862958559344\\
59	-4.0277436969863\\
60	11.0641389560624\\
61	0\\
62	-13.0641314856621\\
63	-18.9361421148615\\
64	-0.293698774035683\\
65	0\\
66	-0.0834590574924204\\
67	0\\
68	0\\
69	0\\
70	0\\
71	0\\
72	0\\
73	20\\
74	8.8055155336003\\
75	-6.93140390328437\\
76	-5.81580269674305\\
77	10.5454686579548\\
78	0\\
79	0\\
80	2.60824967028942\\
81	0\\
82	0\\
83	0\\
84	0\\
85	0\\
86	0\\
87	0\\
88	-14.0417029705714\\
89	-14.4868035228653\\
90	-7.14616453168242\\
91	0\\
92	0\\
93	0\\
94	0\\
95	0\\
96	0\\
97	0\\
98	1.3236522988791e-11\\
99	20\\
100	14.2956151141956\\
101	7.14301998543976\\
102	-4.2810199829546e-12\\
103	0\\
104	0\\
105	0\\
106	0\\
107	0\\
108	0\\
109	-1.32387434348402e-11\\
110	-20\\
111	-14.2862017628232\\
112	-7.14286516056745\\
113	0\\
114	0\\
115	0\\
116	0\\
117	0\\
118	0\\
119	0\\
120	1.32400757024698e-11\\
121	20\\
122	14.2857382871044\\
123	7.14285753762325\\
124	-4.2810199829546e-12\\
125	0\\
126	0\\
127	0\\
128	0\\
129	0\\
130	0\\
};
\addlegendentry{Agent 1}

\addplot [color=mycolor2, line width=1.0pt]
  table[row sep=crcr]{%
1	-20\\
2	-4.16329528205008\\
3	-1.7008607887207\\
4	-8.09015829662922\\
5	-13.3425005643816\\
6	4.44089209850063e-16\\
7	20\\
8	20\\
9	20\\
10	20\\
11	15.1092999842337\\
12	10.4158408206221\\
13	20\\
14	0\\
15	0\\
16	0\\
17	0\\
18	0\\
19	0\\
20	0\\
21	-7.30095812505169\\
22	-20\\
23	4.8746311054856\\
24	5.15857746904658\\
25	-20\\
26	16.2378956177348\\
27	-8.81710318511371\\
28	-11.9182875628515\\
29	-20\\
30	-0.297261994765623\\
31	-20\\
32	-20\\
33	0\\
34	0\\
35	0\\
36	0\\
37	20\\
38	20\\
39	20\\
40	20\\
41	20\\
42	19.693348207266\\
43	3.51189647940613\\
44	0.207828458927299\\
45	0\\
46	0\\
47	-10.7399796084272\\
48	-13.8834594468062\\
49	-20\\
50	-20\\
51	-20\\
52	-20\\
53	-20\\
54	-14.8186725267747\\
55	-20\\
56	-1.71155790433063\\
57	-17.1498873338332\\
58	0\\
59	0\\
60	0\\
61	0\\
62	0\\
63	13.0294267224859\\
64	20\\
65	20\\
66	20\\
67	20\\
68	20\\
69	19.6495095578124\\
70	1.38726629999743\\
71	0.454414959259104\\
72	0\\
73	4.2090078964592\\
74	0\\
75	6.79698054931729\\
76	19.3333599364074\\
77	11.4966617270021\\
78	0\\
79	0\\
80	0\\
81	0\\
82	-5.62005891321557\\
83	-20\\
84	-20\\
85	-19.2360477895757\\
86	-12.6074541407318\\
87	0\\
88	0\\
89	0\\
90	0\\
91	0\\
92	0\\
93	0\\
94	20\\
95	20\\
96	20\\
97	20\\
98	3.19712705646739\\
99	0\\
100	0\\
101	0\\
102	0\\
103	0\\
104	0\\
105	0\\
106	-20\\
107	-20\\
108	-20\\
109	-20\\
110	-0.841349225390188\\
111	0\\
112	0\\
113	0\\
114	0\\
115	0\\
116	0\\
117	0\\
118	20\\
119	20\\
120	20\\
121	20\\
122	0.221407690888125\\
123	0\\
124	0\\
125	0\\
126	0\\
127	0\\
128	0\\
129	0\\
130	-20\\
};
\addlegendentry{Agent 2}

\addplot [color=mycolor3, line width=1.0pt]
  table[row sep=crcr]{%
1	0\\
2	0\\
3	-19.8254751793535\\
4	-11.9639209392074\\
5	0\\
6	3.54962105076993\\
7	6.08798723822153\\
8	0\\
9	0\\
10	-1.49867123842481\\
11	0\\
12	-10.0630640609814\\
13	0\\
14	0\\
15	7.96818336299717\\
16	20\\
17	20\\
18	5.32637928667821\\
19	-6.07375184001285\\
20	15.5465179595194\\
21	-4.37661919543946\\
22	-1.19141177085112\\
23	20\\
24	0.515980003371001\\
25	3.7265143028344\\
26	6.07054083692067\\
27	0.105165645349448\\
28	0\\
29	2.52479311375515\\
30	10.6639785874135\\
31	9.18450721479829\\
32	10.6221441646528\\
33	0\\
34	0\\
35	0\\
36	0\\
37	0\\
38	20\\
39	3.47490482970898\\
40	9.64850767183722\\
41	-1.89629060259904\\
42	5.0497431257444\\
43	13.3289460169908\\
44	3.01382399720643\\
45	0\\
46	6.45757804024983\\
47	0\\
48	0\\
49	-7.42865875043208\\
50	-20\\
51	-19.5181380057476\\
52	-1.49971855203334\\
53	-4.44089209850063e-15\\
54	20\\
55	0\\
56	-20\\
57	-2.60992231108646\\
58	0\\
59	-11.9538719696807\\
60	0\\
61	0.253596085143119\\
62	0\\
63	0\\
64	12.3327800586445\\
65	20\\
66	0.321303507515423\\
67	-0.399056801907136\\
68	2.4863579735696\\
69	0\\
70	8.25466494359716\\
71	13.6621576014939\\
72	0\\
73	0\\
74	-4.61303959442982\\
75	20\\
76	-7.82012291107094\\
77	2.91533256251569\\
78	20\\
79	20\\
80	0\\
81	0\\
82	2.6146118788165\\
83	-4.2774672692758e-12\\
84	-4.83879603052628e-12\\
85	0\\
86	0\\
87	0\\
88	0\\
89	0\\
90	1.29783676135062e-10\\
91	-20\\
92	-13.0713945580341\\
93	-7.1228847789117\\
94	0\\
95	0\\
96	0\\
97	0\\
98	0\\
99	0\\
100	0\\
101	-3.47430863626281e-11\\
102	20\\
103	14.2259260479008\\
104	7.14187378367473\\
105	0\\
106	-4.84590145788388e-12\\
107	0\\
108	0\\
109	0\\
110	-1.77635683940025e-15\\
111	0\\
112	-1.76908793942792e-11\\
113	-20\\
114	-14.2827705524166\\
115	-7.14280872619144\\
116	0\\
117	0\\
118	0\\
119	0\\
120	0\\
121	0\\
122	0\\
123	-2.81828540054384e-11\\
124	20\\
125	14.2855693481139\\
126	7.14285475901162\\
127	0\\
128	-4.84590145788388e-12\\
129	0\\
130	0\\
};
\addlegendentry{Agent 3}

\end{axis}

\end{tikzpicture}%
    \label{fig:uk1per}
  \end{subfigure}
  \begin{subfigure}{0.245\textwidth} 
    \definecolor{mycolor1}{rgb}{0.00000,0.44700,0.74100}%
\definecolor{mycolor2}{rgb}{0.85000,0.32500,0.09800}%
\definecolor{mycolor3}{rgb}{0.92900,0.69400,0.12500}%
\begin{tikzpicture}

\begin{axis}[%
width=1.3in,
height=1.16in,
at={(0.in,0.in)},
scale only axis,
xmin=0,
xmax=130,
xlabel style={font=\color{white!15!black},yshift=.2cm},
xlabel={Time ($10^{-1}$s)},
ymin=-20.5,
ymax=20.5,
ylabel style={font=\color{white!15!black},yshift=-0.5cm},
ylabel={$(u_i)_2$},
axis background/.style={fill=white},
xmajorgrids,
ymajorgrids,
legend columns=2,
legend style={font=\tiny,at={(0.05,0.0)},opacity=0.6, anchor=south west, legend cell align=left, align=left, draw=white!15!black, text opacity=1}
]
\addplot [color=mycolor1, line width=1.0pt]
  table[row sep=crcr]{%
1	-20\\
2	-15.2120441185937\\
3	-6.48157487182918\\
4	0\\
5	-2.99517034883055\\
6	0\\
7	0\\
8	0\\
9	8.98686701890256\\
10	20\\
11	20\\
12	10.5537308390122\\
13	20\\
14	0\\
15	0\\
16	0\\
17	0\\
18	0\\
19	0\\
20	2.60672132825465\\
21	12.4971684054409\\
22	15.097497719762\\
23	0\\
24	0\\
25	8.70847636239584\\
26	0.214634884954042\\
27	-20\\
28	-20\\
29	0.170846187174902\\
30	8.75754571934521\\
31	0\\
32	-5.10160621823161\\
33	-7.04264356379632\\
34	-3.15328124100822\\
35	-0.478164288327913\\
36	-3.35889316343994\\
37	-16.6977709720305\\
38	-19.4311611594982\\
39	-2.38561225400866\\
40	-3.72697544180716\\
41	0\\
42	20\\
43	20\\
44	9.47622409557976\\
45	11.1349303361117\\
46	8.02021054594661\\
47	0.263793068354108\\
48	0.81877490697946\\
49	0.575318867538899\\
50	4.38943158651338\\
51	0.411439737598421\\
52	9.8246752763641\\
53	6.42783615684104\\
54	1.62576101738005\\
55	9.406946206986\\
56	17.3831607059713\\
57	14.3905538950397\\
58	-18.8426663844439\\
59	-19.2442910323488\\
60	9.31242672493681\\
61	-0.949968396162148\\
62	-0.26076233414642\\
63	-5.32133261944621\\
64	-7.83202034196476\\
65	-1.29257355190445\\
66	-5.53127259761529\\
67	-3.39720110679882\\
68	-2.37512018936683\\
69	0\\
70	0\\
71	-10.3003477788453\\
72	0\\
73	0\\
74	0\\
75	-1.86517468137026e-14\\
76	3.04099085340567\\
77	17.4755822539874\\
78	20\\
79	10.8566564767791\\
80	1.92869446182571\\
81	0\\
82	0\\
83	0\\
84	0\\
85	0\\
86	0\\
87	0\\
88	-18.6288886900002\\
89	-14.486803519157\\
90	-7.14616453168242\\
91	0\\
92	0\\
93	0\\
94	0\\
95	0\\
96	0\\
97	0\\
98	1.3236522988791e-11\\
99	20\\
100	14.2956151141956\\
101	7.14301998543976\\
102	-4.2810199829546e-12\\
103	0\\
104	0\\
105	0\\
106	0\\
107	0\\
108	0\\
109	-1.32387434348402e-11\\
110	-20\\
111	-14.2862017628232\\
112	-7.14286516056745\\
113	0\\
114	0\\
115	0\\
116	0\\
117	0\\
118	0\\
119	0\\
120	1.32400757024698e-11\\
121	20\\
122	14.2857382871044\\
123	7.14285753762325\\
124	-4.2810199829546e-12\\
125	0\\
126	0\\
127	0\\
128	0\\
129	0\\
130	0\\
};
\addlegendentry{Agent 1}

\addplot [color=mycolor2, line width=1.0pt]
  table[row sep=crcr]{%
1	-20\\
2	-3.60026025248857\\
3	-0.116141834538819\\
4	-19.1504481462079\\
5	0\\
6	0\\
7	0\\
8	-6.82034841079637\\
9	0\\
10	0.559499831506402\\
11	20\\
12	20\\
13	20\\
14	16.8860300426091\\
15	0\\
16	0\\
17	0\\
18	-16.4797208628038\\
19	-8.07876227039592\\
20	-14.7423037328554\\
21	-20\\
22	-20\\
23	-20\\
24	-4.29638676761378\\
25	-15.1172154975562\\
26	16.3348795759154\\
27	2.37140903882391\\
28	2.71146140354782\\
29	-20\\
30	-20\\
31	-20\\
32	-11.0761823026941\\
33	0\\
34	0\\
35	0\\
36	0\\
37	9.0179592403209\\
38	13.7632533424039\\
39	20\\
40	20\\
41	20\\
42	8.64115468273938\\
43	10.483122484986\\
44	0\\
45	0\\
46	0\\
47	-2.21384078031406\\
48	-19.9657404782192\\
49	-20\\
50	-20\\
51	-20\\
52	-9.67503250832388\\
53	20\\
54	0.725183124571534\\
55	-20\\
56	-20\\
57	-15.3167040574657\\
58	-12.9672828506634\\
59	-0.87197710409227\\
60	0\\
61	0\\
62	0\\
63	0\\
64	20\\
65	20\\
66	20\\
67	20\\
68	20\\
69	20\\
70	20\\
71	20\\
72	8.12526670673358\\
73	3.98820287385726\\
74	1.63205077369639\\
75	0\\
76	0\\
77	0\\
78	0\\
79	-11.6305275401153\\
80	-20\\
81	-20\\
82	-20\\
83	-20\\
84	-19.099071934046\\
85	0\\
86	-10.5281565318919\\
87	0\\
88	0\\
89	0\\
90	0\\
91	0\\
92	0\\
93	0\\
94	20\\
95	20\\
96	20\\
97	20\\
98	4.88503437834305\\
99	0\\
100	0\\
101	0\\
102	0\\
103	0\\
104	0\\
105	0\\
106	-20\\
107	-20\\
108	-20\\
109	-20\\
110	-1.28553536271379\\
111	0\\
112	0\\
113	0\\
114	0\\
115	0\\
116	0\\
117	0\\
118	20\\
119	20\\
120	20\\
121	20\\
122	0.338298779637338\\
123	0\\
124	0\\
125	0\\
126	0\\
127	0\\
128	0\\
129	0\\
130	-20\\
};
\addlegendentry{Agent 2}

\addplot [color=mycolor3, line width=1.0pt]
  table[row sep=crcr]{%
1	20\\
2	3.01502724137476\\
3	0\\
4	0\\
5	-2.81009178470413\\
6	-10.9064970019972\\
7	-19.1439549384086\\
8	-20\\
9	-16.4561488669492\\
10	0\\
11	0\\
12	0\\
13	20\\
14	0\\
15	3.66442147205816\\
16	-6.57111909538798\\
17	-11.4044864515745\\
18	0\\
19	-4.61505776504055\\
20	20\\
21	20\\
22	12.4943462931697\\
23	13.1208627191705\\
24	2.68617776317114\\
25	0\\
26	0\\
27	-5.26061974004019\\
28	-20\\
29	0\\
30	-1.40867634974991\\
31	0\\
32	-9.99509564207326\\
33	-15.2805442301221\\
34	0\\
35	0\\
36	0\\
37	20\\
38	20\\
39	20\\
40	20\\
41	-6.59040674072457\\
42	4.41277220576395\\
43	0\\
44	0\\
45	1.04426371631575\\
46	0\\
47	0\\
48	-6.7904816325485\\
49	0\\
50	-10.643960047552\\
51	1.23142557228372\\
52	3.83673790876014\\
53	9.49406200372323\\
54	8.88178419700125e-16\\
55	-20\\
56	-20\\
57	0\\
58	-7.78531151668479\\
59	-1.7835463129121\\
60	-8.85029672002875\\
61	0\\
62	0\\
63	0\\
64	0.46138712585375\\
65	6.76927511898072\\
66	14.218247089909\\
67	-18.838542292433\\
68	-2.27827516916008\\
69	-5.40656734010554\\
70	11.9854408018728\\
71	5.24128820149505\\
72	-19.1436398148653\\
73	2.03694682750211\\
74	9.83615729932353\\
75	14.7440761990794\\
76	-12.1313814418681\\
77	-0.819273114514847\\
78	7.96666238973557\\
79	20\\
80	-14.8240436304332\\
81	20\\
82	11.3350820178848\\
83	0\\
84	-4.84945417156268e-12\\
85	0\\
86	-1.0842021724855e-19\\
87	0\\
88	1.77635683940025e-15\\
89	0\\
90	2.81873344234367e-11\\
91	-20\\
92	-15.1439411637834\\
93	-7.15697271650194\\
94	0\\
95	0\\
96	0\\
97	0\\
98	0\\
99	0\\
100	0\\
101	5.04710605963273e-11\\
102	20\\
103	14.327969938395\\
104	7.14355213714665\\
105	0\\
106	-4.84945417156268e-12\\
107	0\\
108	0\\
109	0\\
110	0\\
111	0\\
112	1.83574674403392e-11\\
113	-20\\
114	-14.2877947847615\\
115	-7.14289136157803\\
116	0\\
117	0\\
118	0\\
119	0\\
120	0\\
121	0\\
122	0\\
123	1.44192338431823e-11\\
124	20\\
125	14.2858167211319\\
126	7.14285882765053\\
127	0\\
128	-4.84590145788388e-12\\
129	0\\
130	0\\
};
\addlegendentry{Agent 3}

\end{axis}

\end{tikzpicture}%
    \label{fig:uk2per}
  \end{subfigure}
  \caption{(Left) Nominal and perturbed trajectories. (Right) Accelerations under random perturbations in the velocity terms.}
  \label{fig:main2}
\end{figure*}

\section{Conclusion}\label{sec:conclusion}

We have proposed a receding horizon control strategy to address the infinite-horizon synthesis problem of a multi-agent system under recurring STL specifications. To ensure recursive feasibility and guarantee the integrity of the proposed control scheme we have introduced additional constraints and terminal conditions that are easy to construct and intuitive for schemes under recurring tasks. By decomposing the global optimization problem into agent-level programs we implement a scheduling policy that enables individual agents to optimize their control actions sequentially. This approach yields a distributed control strategy that enables online operation, while preserving recursive feasibility. Our method hinges on the feasibility of an initial iteration carried out at the MAS level. Inspired by the proposed scheduling policy and the decomposition approach, we aim to tackle the problem in a fully distributed manner in future work.




\appendices
\section*{Appendix I: Proof of Theorem \ref{thm:recurs_feasibility}}\label{ap:proofThm1}

We show recursive feasibility for $t>N$. The proof for $t\in\N_{[0,N]}$ is similar. Suppose that \eqref{eq:receding_horizon_relax_MAS} is feasible at time $t>N$, with optimal input sequence $\bm{u}^\ast(t_{0:N})$ and optimal state trajectory $\bm{x}^\ast(t_{0:N+1})$. By applying $u(t) = u^\ast(t_0)$, the MAS is driven to $x(t+1)=f(x(t),u(t))=x^\ast(t_1)$. Since $x^\ast(t_N)\in\cC_1(x(t))$, there is an admissible input, $\hat{u}(t+N)$, such that $f(x^\ast(t_N),\hat{u}(t+N)) = x(t)$. At time $t+1$, we design admissible input sequence, $\bm{u}((t+1)_{0:N}) = (\bm{u}^\ast(t_{1:N-1}),\hat{u}(t+N),u^\ast(t_0))$, which yields $\bm{x}((t+1)_{0:N+1})= (x(t+1),\bm{x}^\ast(t_{2:N}),x(t),x(t+1))$, which we check for feasibility. Constraints in \eqref{eq:history1_MAS}-\eqref{eq:history4_MAS}, which, for time $t+1$, are written as $(x(t-N+2),x(t-N+3),\ldots,\bm{x}^\ast(t_{1:2})) \models \phi$, $(x(t-N+3),x(t-N+4),\ldots,\bm{x}^\ast(t_{2:3})) \models \phi$, $\ldots$, $(x(t), \bm{x}^\ast(t_{1:N})) \models \phi$, are feasible by constraints in \eqref{eq:history1_MAS}-\eqref{eq:constraint_phi_MAS} corresponding to time $t$. Constraint in \eqref{eq:constraint_phi_MAS} for time $t+1$, namely, $\bm{x}((t+1)_{0:N}) \models \phi$, which is written as $(\bm{x}^\ast(t_{1:N}), x(t)) \models \phi$, is identical to the constraint in \eqref{eq:con6_MAS} at time $t$, and, hence, is feasible. Constraints in \eqref{eq:con1_MAS}-\eqref{eq:con5_MAS} at time $t+1$ are written as $(\bm{x}^\ast(t_{2:N+1}),x^\ast(t_1))\models \phi$, $(\bm{x}^\ast(t_{3:N+1}),\bm{x}^\ast(t_{1:2}))\models \phi$, $\ldots$, $\bm{x}^\ast(t_{1:N+1})\models\phi$, which are identical to constraints in \eqref{eq:con1_MAS}-\eqref{eq:con5_MAS} at time $t$, with a different order. Due to the choice of $x((t+1)_{N})=x(t)$, constraints in \eqref{eq:con6_MAS}-\eqref{eq:con7_MAS} for the time step $t+1$ are identical to constraints in \eqref{eq:con1_MAS}-\eqref{eq:con5_MAS} at time $t+1$, and, hence, are feasible. Lastly, since $x((t+1)_{N+1}) = x(t+1)$, $x((t+1)_{N})\in\cC_1(x(t+1))$ implying that constraint in \eqref{eq:varying_terminal_constr_MAS} is feasible. This completes the proof. 

\bibliographystyle{IEEEtran} 
\balance
\bibliography{IEEEabrv,biblio}

\end{document}